\documentclass[12pt, reqno]{amsart}

\usepackage{graphics, stackrel}
\usepackage{amsmath, amssymb, amsthm, amsfonts}
\usepackage{graphicx}
\usepackage{verbatim}
\usepackage{amsfonts}
\usepackage{natbib}
\usepackage{enumitem}
\usepackage{epigraph}
\usepackage{booktabs}

\usepackage[xcharter]{newtxmath}
\usepackage{caption}
\usepackage{subcaption}

\usepackage{fancyvrb}
\usepackage[dvipsnames,svgnames,table]{xcolor}
\usepackage{mdwlist}

\usepackage[breaklinks=true, colorlinks=true, citecolor=Blue, linkcolor=blue]{hyperref}
\newcommand\myshade{85}
\colorlet{mylinkcolor}{violet}
\colorlet{mycitecolor}{Brown}
\colorlet{myurlcolor}{violet}
\hypersetup{
  linkcolor  = mylinkcolor!\myshade!black,
  citecolor  = mycitecolor!\myshade!black,
  urlcolor   = myurlcolor!\myshade!black,
  colorlinks = true,
}

\usepackage{enumitem}
\setlist[enumerate]{itemsep=2pt,topsep=3pt}
\setlist[itemize]{itemsep=2pt,topsep=3pt}
\setlist[enumerate,1]{label={\upshape (\roman*)}}

\usepackage{bbm}

\usepackage[left=1.25in, 
            right=1.25in, 
            top=1.0in, 
            bottom=1.15in, 
            includehead, includefoot]{geometry}

\newcommand{\navy}[1]{\textit{\textcolor{Blue}{#1}}}

\renewcommand{\leq}{\leqslant}
\renewcommand{\geq}{\geqslant}

\usepackage[ruled, linesnumbered]{algorithm2e}

\DeclareMathOperator*{\argmin}{arg\,min}
\DeclareMathOperator*{\argmax}{arg\,max}

\newcommand{\setntn}[2]{ \{ #1 : #2 \} }

\newcommand{\1}{\mathbbm 1}

\newcommand{\asim}{\stackrel { a } {\sim} }

\newcommand*\diff{\mathop{}\!\mathrm{d}}

\newcommand{\tmax}{T_\vartriangle}
\newcommand{\tmin}{T_\triangledown}
\newcommand{\vmax}{v_\vartriangle}
\newcommand{\vmin}{v_\triangledown}
\newcommand{\Hmax}{H_\vartriangle}
\newcommand{\Hmin}{H_\triangledown}

\newcommand{\htmax}{\hat{T}_\vartriangle}
\newcommand{\htmin}{\hat{T}_\triangledown}
\newcommand{\hvmax}{\hat{v}_\vartriangle}
\newcommand{\hvmin}{\hat{v}_\triangledown}

\newcommand{\aA}{\mathscr A}

\newcommand{\eE}{\mathcal E}

\newcommand{\RR}{\mathbbm R}

\newcommand{\NN}{\mathbbm N}

\newcommand{\EE}{\mathbbm E}

\newcommand{\Asf}{\mathsf A}

\newcommand{\Gsf}{\mathsf G}

\newcommand{\Xsf}{\mathsf X}

\newcommand{\Ysf}{\mathsf Y}
\newcommand{\Zsf}{\mathsf Z}

\renewcommand{\phi}{\varphi}
\renewcommand{\epsilon}{\varepsilon}

\theoremstyle{plain}
\newtheorem{theorem}{Theorem}[section]

\newtheorem{lemma}[theorem]{Lemma}
\newtheorem{proposition}[theorem]{Proposition}

\theoremstyle{definition}

\newtheorem{example}{Example}[section]
\newtheorem{remark}{Remark}[section]

\begin{document}

\title{Isomorphic Dynamic Programs}

\author{John Stachurski}
\address{National Graduate Institute for Policy Studies}
\email{john.stachurski@gmail.com}

\author{Junnan Zhang}
\address{Center for Macroeconomic Research and Wang Yanan Institute for
  Studies in Economics at Xiamen University, China}
\email{zhangjunnan1224@gmail.com}

\thanks{The authors thank Yuchao Li, Shu Hu, and Jingni Yang for valuable
comments and suggestions.  Financial support from Schmidt Futures is gratefully acknowledged.
The second author acknowledges support from NSFC Project 72303190.}

\begin{abstract}
    We study relationships between dynamic programs by applying conjugacy
    methods from dynamical systems theory.  When two dynamic programs are
    connected by an order isomorphism, we show that optimality properties
    transmit from one formulation to the other.  We apply these results to
    Epstein--Zin preferences with time preference shocks, obtaining a sharp
    characterization of when optimality holds.  We also show that
    multiplicative Kreps--Porteus preferences and risk-sensitive preferences
    are isomorphic, so that well-known results for the latter carry over to
    the former.  Finally, we demonstrate how isomorphic transformations can
    improve the numerical accuracy of value function approximations, with
    gains of two orders of magnitude in a multisector real business cycle model.
\end{abstract}

\date{\today}

\maketitle

\section{Introduction}

Dynamic programming is a major field of optimization with applications across
a wide spectrum of scientific domains, including economics and finance (see,
e.g., \cite{bellman1957dynamic}, \cite{stokey1989recursive},
\cite{puterman2005markov}, \cite{hernandez2012discrete},
\cite{bauerle2011markov}, \cite{bertsekas2012dynamic}). In recent years,
variations of the standard model have flourished.  This is particularly true in
economics and finance, as applied researchers
extend dynamic models to get closer to the data. Some of these variations add
features to the standard framework, including Epstein--Zin preferences,
risk-sensitive preferences, adversarial agents, ambiguity, and hyperbolic
discounting.\footnote{See, for example, \cite{epstein1989risk},
  \cite{cagetti2002robustness}, \cite{hansen2011robustness},
  \cite{bauerle2018stochastic}, \cite{fedus2019hyperbolic},
  \cite{marinacci2019unique}, \cite{gao2021robust}, or \cite{de2022dynamic}.}
Other variations take an existing model and rearrange the structure and
timing. These variations include the Q-factor and exponential risk-sensitive
Q-factor formulations from Q-learning, and the expected value function models and
integrated value function models from structural estimation.\footnote{See,
  for example, \cite{kochenderfer2022algorithms}, \cite{rust1994structural},
  \cite{fei2021exponential}, or \cite{kristensen2021solving}.}

Given this proliferation of variations and modifications, one significant
issue is that each new formulation typically requires its own bespoke analysis
to establish desirable optimality properties (e.g., existence of optimal
policies, Bellman's principle of optimality, or convergence of Howard policy
iteration). In this paper, we address this issue from a new direction, by borrowing ideas from the
field of dynamical systems. Since each policy operator defines a
discrete-time dynamical system, tools from that field become directly
applicable. In particular, we draw on notions of conjugacy. In
dynamical systems and chaos theory, topological conjugacy is used to identify
families of dynamical systems that have ``equivalent'' or related dynamics
(see, e.g., \cite{smale1967differentiable} or \cite{arnold2012geometrical}).
These ideas have proved enormously useful for categorizing different types of
systems and allowing researchers to make deductions about a particular system
from their knowledge of other related systems.\footnote{For example,
  \cite{kennedy2008chaotic}, \cite{gardini2009forward}, and
  \cite{raines2012fixed} characterize chaotic dynamical systems in economic
  models utilizing the shift homeomorphism. In a similar vein,
  \cite{flynn2022macroeconomics} establish topological conjugacy between
  equilibrium economic dynamics and dynamical systems with known chaotic
  behaviors. Also see \cite{battaglini2021chaos} and
  \cite{deng2022continuous}.}

Specifically, we transfer these ideas to dynamic programming by studying
order-theoretic conjugacy relationships between the policy operators that
define a dynamic program.  We call two dynamic programs linked in this way
\emph{isomorphic}; under an order-reversing link, we call them
\emph{anti-isomorphic}.  Our main result shows that verifying this single
relationship is enough to transfer an entire package of optimality properties
from one program to the other: existence and uniqueness of the value
function, Bellman's principle of optimality, existence of optimal policies,
and convergence of Howard policy iteration.  Solving one program therefore
delivers a complete solution to the other, including recovery of the value
function and optimal policies.

The benefits of this approach are both theoretical and numerical. On the
theoretical side, dynamic programs that are not easily analyzed can often be
connected, via an isomorphism, to a program whose optimality properties are
already established or more readily verified; the isomorphism then allows
those properties to carry over to the original problem. We illustrate this in
two settings that are not covered by standard contraction-based methods.
First, for a dynamic program with Epstein--Zin utility and time preference
shocks, we obtain an exact characterization of when optimality holds. Second,
we show that multiplicative Kreps--Porteus preferences are isomorphic to
risk-sensitive preferences, so that the well-developed contraction theory for
the latter yields optimality results for the former. On the numerical side,
within an isomorphism class one is free to work with whichever representative
is best suited to computation. Exploiting this flexibility in the multisector
real business cycle model of \cite{long1983real}, we obtain solutions that
are, on average, two orders of magnitude more accurate than solving the
original problem, by choosing a representative within the isomorphism class
whose value function has lower curvature.

This paper builds on the abstract dynamic programming framework of
\cite{sargent2025partially}. In that paper, individual dynamic programs are
represented as families of policy operators acting on a partially ordered space.
When dynamic programs are viewed in this way, studying conjugacy relationships
between related families of policy operators becomes a natural idea, since each
policy operator identifies a discrete time dynamical system.  The main deviation
from traditional dynamical systems methodology is that we use order-theoretic
conjugacy relationships, rather than topological ones, since the key objects
in optimization are order-theoretic and preserved by this relation.

This paper contributes to a large and growing literature that studies the
optimality properties of dynamic programs with nonstandard features, such as
recursive utility \citep{marinacci2010unique, bloise2024not,
rincon2024existence}, risk sensitive preferences \citep{bauerle2018stochastic},
hyperbolic discounting \citep{bauerle2021stochastic, balbus2022time}, and
Q-factors \citep{ma2022unbounded}. For example, our work relates to the
literature on existence and uniqueness in dynamic programs with recursive
preferences, such as \cite{bloise2018convex} and \cite{bloise2024not}. Our
approach complements this literature in two ways. First, rather than focusing on
specific models, we propose a general framework that, under appropriate
transformations, can simplify the analysis of such dynamic programs. Second, as
an application of our theory, we provide exact necessary and sufficient
conditions for optimality in dynamic programs with Epstein--Zin preferences and
state dependent discounting. 

% The two approaches are complementary: in settings where direct verification
% of uniqueness is delicate, an isomorphism can transform the problem into one
% where the conditions of \citet{bloise2024blame} or
% \citet{christensen2022existence} are straightforward to verify, as we
% illustrate in Section~\ref{sec:ez}.

While not the primary focus of this paper, our work also connects to the
literature on computational methods for solving dynamic economic models.
These include traditional methods such as projection and perturbation
\citep{foerster2016perturbation, bayer2020solving}, endogenous grid method \citep{carroll2006method,
  barillas2007generalization}, adaptive sparse grids \citep{brumm2017using},
and deep learning \citep{azinovic2022deep, maliar2021deep}. A notable feature
of our framework is that it can complement existing computation methods: any
suitable method can be applied to an isomorphic dynamic program to improve
numerical performance compared to solving the original problem.

The structure of the paper is as follows.  Section~\ref{s:prelim} introduces
order conjugacy and discusses its properties. Sections~\ref{s:adps}--\ref{s:aop}
recall some definitions and optimality results for ``abstract'' dynamic
programs, which are well-suited to our dynamical systems perspective.
Section~\ref{s:iso} introduces isomorphisms between abstract dynamic
programs and shows how isomorphisms preserve optimality properties.
Section~\ref{s:app} provides applications and Section~\ref{sec:concl} concludes.

\section{Preliminaries}\label{s:prelim}

This section introduces the key order-theoretic concepts needed for our
analysis of abstract dynamic programs. Additional definitions of posets,
dynamical systems, and related background are given in
Appendix~\ref{app:order}.  Throughout, if $A$ and $B$ are sets, then the
symbol $B^A$ indicates the set of all maps from $A$ to $B$.

A \navy{partially ordered set} (poset) is a pair $V = (V, \preceq)$ where
$\preceq$ is reflexive, antisymmetric, and transitive.  Given posets $(V,
\preceq)$ and $(W, \triangleleft)$, a function $F \colon V \to W$ is called
\navy{order-preserving} if $v \preceq w$ implies $Fv \triangleleft Fw$, and
\navy{order-reversing} if $v \preceq w$ implies $Fw \triangleleft Fv$.  A
bijection $F \colon V \to \hat V$ is an \navy{order isomorphism} if both $F$
and $F^{-1}$ are order-preserving, and an \navy{order anti-isomorphism} if
both $F$ and $F^{-1}$ are order-reversing.

Let $(V, S)$ be a dynamical system (a set $V$ with a self-map $S$) where $V$
is a poset and $S$ has a unique fixed point $\bar v$.  Following
\cite{sargent2025partially}, $(V, S)$ is called \navy{upward stable} if $v
\preceq S\,v$ implies $v \preceq \bar v$, \navy{downward stable} if $S\,v
\preceq v$ implies $\bar v \preceq v$, and \navy{order stable} if both hold.

Figure~\ref{f:up_down_stable} gives an illustration of an order stable map
$S$ on $V = [0,1]$. All points mapped up by $S$ lie below its unique fixed
point, while all points mapped down by $S$ lie above its fixed point.

\begin{figure}
    \centering
    \scalebox{0.45}{\includegraphics{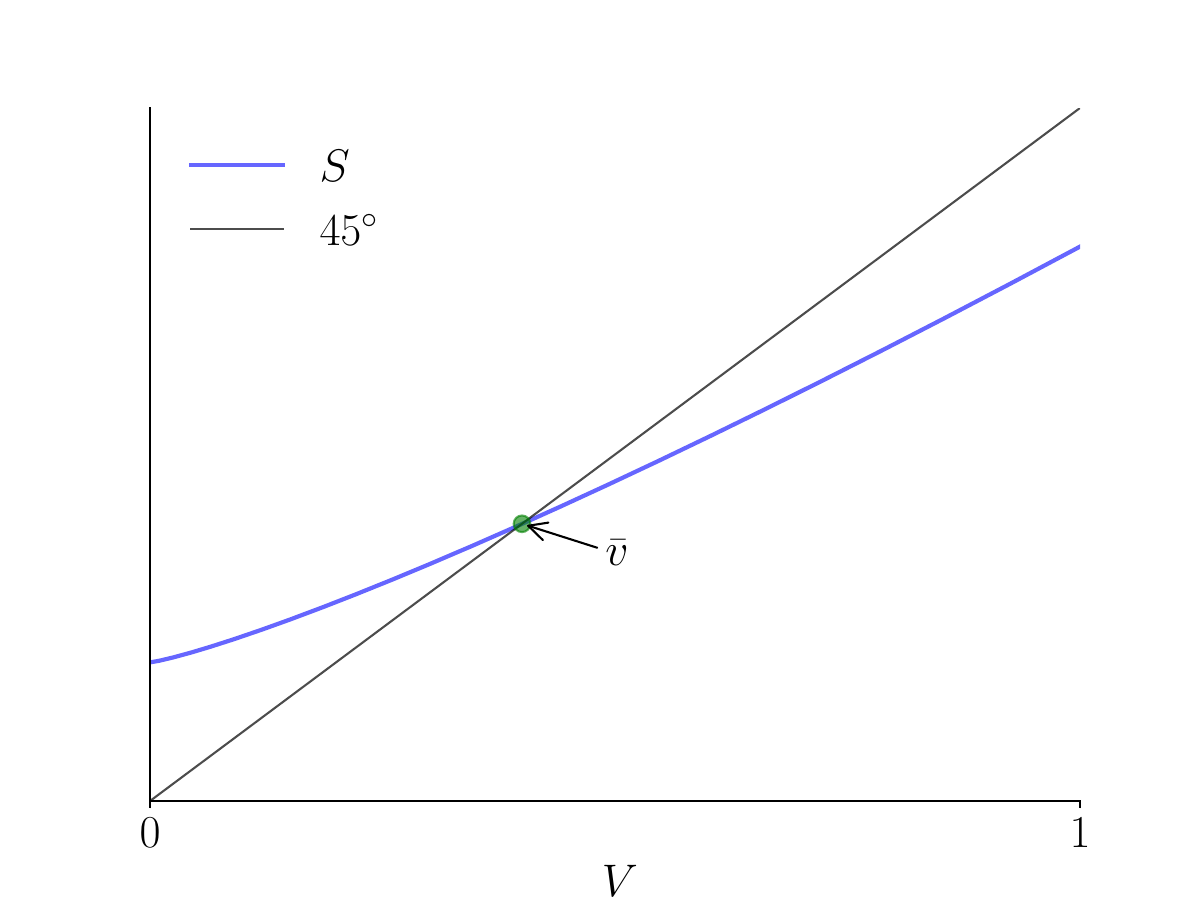}}
    \caption{\label{f:up_down_stable} An order stable map $S$ on $[0,1]$}
\end{figure}

Intuitively, order stability says that the fixed point $\bar v$ of $S$
separates $V$ into two regions: elements that $S$ maps upward all lie below
$\bar v$, and elements that $S$ maps downward all lie above $\bar v$. This is
an order-theoretic analogue of global stability. (A dynamical system $(V, S)$
on a metric space is called \navy{globally stable} if $S$ has a unique fixed
point $\bar v$ and $S^k v \to \bar v$ for all $v \in V$.) Whereas global
stability requires a topology and ensures convergence of trajectories, order
stability requires only a partial order and is not as strong. In fact, order
stability is implied by global stability when $S$ is order-preserving and the
space $V$ has both a topology \emph{and} a partial order:

\begin{example}\label{eg:pspace}
    Let $(V, \preceq, \rho)$ be a partially ordered space (i.e., a metric space
    $(V, \rho)$ where $\preceq$ is closed under limits). If $(V, S)$ is
    globally stable and $S$ is order-preserving, then $(V, S)$ is order
    stable. To see this, let $\bar v$ be the unique fixed point of $S$ in $V$.
    Upward stability holds because if $v \in V$ and $v \preceq S \, v$, then,
    iterating on this inequality and using the fact that $S$ is
    order-preserving,  we have $v \preceq S^k \, v$ for all $k \in \NN$.
    Applying global stability and taking the limit gives $v \preceq \bar
    v$.  Hence upward stability holds. The proof of downward stability is
    similar.
\end{example}

Two dynamical systems $(V, S)$ and $(\hat V, \hat S)$ on posets are called
\navy{order conjugate} under $F$ when they are conjugate (i.e., $F \circ S =
\hat S \circ F$ for some bijection $F$) and $F$ is an order isomorphism.
Order conjugacy is an equivalence relation that preserves order stability:
if $(V, S)$ and $(\hat V, \hat S)$ are order conjugate, then one is order
stable if and only if the other is (Lemma~\ref{l:ocos} in
Appendix~\ref{app:oc}).

\section{Abstract Dynamic Programs}\label{s:adps}

We follow the abstract dynamic programming framework of
\cite{sargent2025partially}, which was partly inspired by 
\cite{bertsekas2022abstract}.  For us, the benefit of the framework in \cite{sargent2025partially} is that
dynamic programs are represented as families of policy operators. This abstract
representation allows us to apply the notion of order conjugacy to study
relationships between dynamic programs.  

\subsection{Definition and Examples}

\cite{sargent2025partially} define an \navy{abstract dynamic program} (ADP)
to be a pair $\aA = (V, \{T_\sigma\}_{\sigma \in \Sigma})$, where 
\begin{enumerate}
    \item $V = (V, \preceq)$ is a partially ordered set and
    \item $\{T_\sigma\}_{\sigma \in \Sigma}$ is a family of order-preserving
        self-maps on $V$, indexed by $\sigma \in \Sigma$.
\end{enumerate}
Elements of the index set $\Sigma$ are referred to as \navy{policies} and
elements of $\{T_\sigma \}_{\sigma \in \Sigma}$ are called \navy{policy operators}.   When $\Sigma$
is understood, we often write $\{T_\sigma \}_{\sigma \in \Sigma}$ as $\{T_\sigma
\}$. In all applications, the significance of each policy operator $T_\sigma$ is that its
fixed point, denoted below by $v_\sigma$, represents the lifetime value (or cost) of following
policy $\sigma$.   

\begin{example}[MDPs]\label{eg:mdp}
    Consider a \navy{Markov decision process} (MDP; see, e.g.,
    \cite{puterman2005markov}) where the aim is to maximize $\EE \sum_{t \geq 0}
    \beta^t r(X_t, A_t)$ when $X_t$ takes values in finite set $\Xsf$ (the state
    space), $A_t$ takes values in finite set $\Asf$ (the action space), $\Gamma$
    is a nonempty correspondence from $\Xsf$ to $\Asf$, the set
        $\Gsf \coloneq \setntn{(x, a) \in \Xsf \times \Asf}{a \in \Gamma(x)}$
    denotes feasible state-action pairs, $r$ is a reward function defined on
    $\Gsf$, $\beta \in (0,1)$ is a discount factor, and $P \colon \Gsf \times
    \Xsf \to [0,1]$ provides transition probabilities. The Bellman equation for this problem is 
    \begin{equation}\label{eq:adpbell}
        v(x) = \max_{a \in \Gamma(x)} 
        \left\{
            r(x, a) + \beta \sum_{x'} v(x') P(x, a, x')
        \right\}
        \qquad (x \in \Xsf).
    \end{equation}
    The set of feasible policies is
        $\Sigma := 
            \setntn{\sigma \in \Asf^\Xsf}
            {\sigma(x) \in \Gamma(x) \text{ for all } x \in \Xsf}$.
    We combine $\RR^\Xsf$ (the set of all real-valued functions on $\Xsf$)
    with the pointwise partial order $\leq$ and, for $\sigma \in
    \Sigma$ and $v \in \RR^\Xsf$, define the MDP policy operator
    \begin{equation}\label{eq:tsig_mdp}
        (T_\sigma \, v)(x) 
            = r(x, \sigma(x)) + \beta \sum_{x'} v(x') P(x, \sigma(x), x')
        \qquad (x \in \Xsf).
    \end{equation}
    The pair $(\RR^\Xsf, \{T_\sigma\})$ is an ADP.
\end{example}

\begin{example}[Risk-sensitive MDPs]\label{eg:risksens}
    Risk-sensitive MDPs \citep{howard1972risk} modify the MDP model in
    Example~\ref{eg:mdp} so that the policy operators take the form 
    \begin{equation*}
        (T_\sigma^\theta \, v)(x) 
        = r(x, \sigma(x)) + \frac{\beta}{\theta}
            \ln \left[ 
                \sum_{x'} \exp(\theta v(x')) P(x, \sigma(x), x')
            \right]
    \end{equation*}
    where $\theta$ is some fixed value in $\Theta := \RR \setminus \{0\}$.
    The pair $(\RR^\Xsf, \{T_\sigma^\theta\})$ is an ADP.
\end{example}

Further examples of ADPs, including $Q$-factor and risk-sensitive $Q$-factor
formulations from the reinforcement learning literature, are given in
Appendix~\ref{app:qfac}.

\subsection{Lifetime Values}\label{ss:lifeval}

The objective of dynamic programming is to maximize or minimize
lifetime value.  Following \cite{sargent2025partially}, we identify the lifetime
value of any given policy $\sigma$ as the unique fixed point of $T_\sigma$,
whenever it exists. When it does exist, we denote this fixed point by $v_\sigma$ and call it
the \navy{$\sigma$-value function}. 

\begin{example}
    Consider the MDP setting of Example~\ref{eg:mdp}.
    Let $r_\sigma$ and $P_\sigma$ be defined by 
    \begin{equation}\label{eq:rps}
        r_\sigma(x) := r(x, \sigma(x))
        \quad \text{and} \quad
        (P_\sigma v)(x) := \sum_{x'}v(x')P(x, \sigma(x), x').
    \end{equation}
    The lifetime value of policy $\sigma$ given $X_0 = x$ is $v_\sigma(x) = \EE
    \sum_{t \geq 0} \beta^t r(X_t, \sigma(X_t))$, where $(X_t)_{t \geq 0}$ is a Markov
    chain generated by $P_\sigma$ with initial condition $X_0 = x \in \Xsf$. 
    Pointwise on $\Xsf$, we can express $v_\sigma$ as
    \begin{equation}\label{eq:vsigmdp}
        v_\sigma = \sum_{t \geq 0} (\beta P_\sigma)^t r_\sigma
                = (I-\beta P_\sigma)^{-1} r_\sigma
    \end{equation}
    (see, e.g., \cite{puterman2005markov}, Theorem~6.1.1, or
    \cite{kochenderfer2022algorithms}, Section~7.2). Equivalently, $v_\sigma$
    is the unique solution to the equation $v = r_\sigma + \beta P_\sigma \, v$.
    Inspecting the definition of $T_\sigma$ in \eqref{eq:tsig_mdp}, we see this is
    also equivalent to the statement that $v_\sigma$ is the unique fixed point of
    $T_\sigma$.
\end{example}

Similarly, for the risk-sensitive MDP of Example~\ref{eg:risksens}, the
operator $T_\sigma^\theta$ is a contraction on $\RR^\Xsf$ and its unique
fixed point represents the lifetime value of following $\sigma$.

\subsection{Greedy Policies}

The core idea behind the theory of dynamic programming is that 
an optimal policy can be obtained by choosing actions that solve a
two-period problem involving a Bellman equation \citep{bellman1957dynamic}.
The optimal actions in the two-period problem produce what are typically called
``greedy policies.'' For example, in the context of Example~\ref{eg:mdp}, a
policy $\sigma$ satisfying 
\begin{equation}\label{eq:mdpmvg}
    \sigma(x) \in \argmax
    \left\{
        r(x, a) + \beta \sum_{x'} v(x') P(x, a, x')
    \right\}
    \qquad \text{for all $x \in \Xsf$}.
\end{equation}
is called a ``greedy policy with respect to $v$.''  See, for example, \cite{kochenderfer2022algorithms}, Section~7.3.

Following \cite{sargent2025partially}, we generalize this idea to the abstract
setting.  In particular, given an ADP $(V, \{T_\sigma\})$ and an element
$v \in V$, a policy $\sigma$ in $\Sigma$ is called 
\begin{itemize}
    \item \navy{$v$-min-greedy} if $T_{\sigma} \, v \preceq T_\tau \, v$ for all
        $\tau \in \Sigma$, and 
    \item \navy{$v$-max-greedy} if $T_{\tau} \, v \preceq T_\sigma \, v$ for all
        $\tau \in \Sigma$. 
\end{itemize}
For example, in the context of Example~\ref{eg:mdp}, with $\leq$ as the
pointwise partial order, a policy $\sigma$ obeying
\eqref{eq:mdpmvg} satisfies $T_\tau \, v \leq T_\sigma \, v$ for
all $\tau \in \Sigma$, and hence is $v$-max-greedy. A $v$-min-greedy policy can
be constructed by replacing $\argmax$ in \eqref{eq:mdpmvg} with $\argmin$.

\subsection{Bellman Operators}

For a generic ADP $(V, \{T_\sigma\})$, we respectively define the \navy{Bellman
min-operator} and the \navy{Bellman max-operator} via 
\begin{equation}\label{eq:tbvemax}
    \tmin \, v := \bigwedge_\sigma T_\sigma \, v 
    \quad \text{and} \quad
    \tmax \, v := \bigvee_\sigma T_\sigma \, v 
\end{equation}
whenever the infimum (resp., supremum) exists. We say that $v \in V$ 
satisfies the \navy{Bellman max-equation} (resp., the \navy{Bellman min-equation})
if it is a fixed point of $\tmax$ (resp., $\tmin$). Notice that $\sigma \in
\Sigma$ is
\begin{enumerate}
    \item $v$-max-greedy if and only if $T_\sigma \, v = \tmax \, v$, and   
    \item $v$-min-greedy if and only if $T_\sigma \, v = \tmin \, v$.
\end{enumerate}

To illustrate, consider the MDP setting of Example~\ref{eg:mdp}.  Traditionally, 
the Bellman operator for this model is given by 
\begin{equation}\label{eq:mdpt}
    (Tv)(x) = \max_{a \in \Gamma(x)} 
    \left\{
        r(x, a) + \beta \sum_{x'} v(x') P(x, a, x')
    \right\}
    \qquad (x \in \Xsf).
\end{equation}
(See, e.g., \cite{puterman2005markov}.)
For the ADP $(\RR^\Xsf, \{T_\sigma\})$ generated by this MDP, this Bellman
operator exactly coincides with the Bellman max-operator $\tmax$ in
\eqref{eq:tbvemax}.  Replacing $\max$ with $\min$ in \eqref{eq:mdpt} produces
the Bellman min-operator from \eqref{eq:tbvemax}.

\subsection{Properties of ADPs}\label{ss:prop}

To obtain optimality results, we need to place structure on ADPs.  Here and
below, we call an ADP $\aA = (V, \{T_\sigma\})$ 
\begin{itemize}
    \item \navy{well-posed} if $T_\sigma$ has a unique fixed point $v_\sigma$ in
        $V$ for each $\sigma \in \Sigma$,
    \item \navy{order stable} if $(V, T_\sigma)$ is order stable for each
        $\sigma \in \Sigma$, 
    \item \navy{max-stable} if $\aA$ is order stable, each $v \in V$ has at
        least one max-greedy policy, and $\tmax$ has at least one fixed point in
        $V$, and
    \item \navy{min-stable} if $\aA$ is order stable, each $v \in V$ has at
        least one min-greedy policy, and $\tmin$ has at least one fixed point in
        $V$.
\end{itemize}

Well-posedness is a minimum regularity condition for ADPs.  Without it, we
cannot be sure that policies have well-defined lifetime values. Well-defined
lifetime values are essential because maximizing (or minimizing) lifetime value
over the set of all policies is the objective of dynamic programming.

% \begin{example}[Continuous time MDPs]
%     We return to the MDP setting of Example~\ref{eg:mdp} but
%     replacing the discount factor $\beta$ with a discount rate $\delta > 0$
%     and the stochastic kernel $P$ with an intensity kernel 
%     $Q$ from $\Gsf \times \Xsf \to \RR$ satisfying $\sum_{x' \in \Xsf} Q(x, a,
%     x') = 0$ for all $(x,a)$ in $\Gsf$ and $Q(x, a, x') \geq 0$ whenever $x
%     \not= x'$. Lifetime value of policy $\sigma$ is given by 
%     %
%     \begin{equation}\label{eq:ctlv}
%         v_\sigma (x) 
%         = \EE_x \int_0^\infty \exp(-\delta t) r(X_t, \sigma(X_t))  \diff t.
%     \end{equation}
%     %
%     where $Q_\sigma(x,x') := Q(x, \sigma(x), x')$ is the infinitesimal generator
%     of the continuous time Markov chain $(X_t)_{t \geq 0}$ and $r_\sigma(x) :=
%     r(x, \sigma(x))$.  The function $v_\sigma$ can alternatively be written as
%     $v_\sigma = (\delta I - Q_\sigma)^{-1} r_\sigma$  (see, e.g.,
%     \cite{guo2009continuous}, Lemma~4.16).  Rearranging this expression
%     and using the fact that $\delta I - Q_\sigma$ is bijective shows that
%     $v_\sigma$ is the unique fixed point of $T_\sigma \, v = r_\sigma +
%     (Q_\sigma + (1 - \delta) I) v$ in $\RR^\Xsf$.  Hence $(\RR^\Xsf,
%     \{T_\sigma\})$ is a well-posed ADP.
% \end{example}
%

Order stability is a natural regularity assumption for ADPs. To understand this,
suppose $\aA$ is well-posed and consider upward stability for an arbitrary
policy operator $T_\sigma$ with fixed point $v_\sigma$.  If $v \preceq T_\sigma
\, v$, then following policy $\sigma$ for one period offers an improvement in
value. Since the problem is stationary, this suggests that following the policy
forever will also be an improvement. Thus, we expect $v \preceq
v_\sigma$, in which case upward stability holds. Intuition for downward
stability is similar.

\begin{example}\label{eg:mdpos}
    Consider the MDP setting of Example~\ref{eg:mdp} and fix $\sigma \in \Sigma$. 
    Recall that the policy operator $T_\sigma$ has
    unique fixed point $v_\sigma := (I - \beta P_\sigma)^{-1} r_\sigma$.
    If $T_\sigma \, v \geq v$, then $r_\sigma + \beta P_\sigma \, v \geq v$ and
    hence $(I - \beta P_\sigma) v \leq r_\sigma$.  Since $(I - \beta
    P_\sigma)^{-1}$ is positive, we get $v \leq v_\sigma$.  Hence $(V,
    T_\sigma)$ is upward stable.  A similar proof shows that $(V, T_\sigma)$ is
    downward stable, and therefore order stable.
\end{example}

\begin{example}
    Consider the ADP $(\RR^\Xsf, \{T_\sigma\})$ from the MDP setting of Example~\ref{eg:mdp}. 
    For this case $\tmax$ is given by \eqref{eq:mdpt}.
    Since $\tmax$ is a contraction on $\RR^\Xsf$, it has a unique fixed point in $\RR^\Xsf$.  
    Since max-greedy policies always exist, and since $(\RR^\Xsf, \{T_\sigma\})$
    is order stable by Example~\ref{eg:mdpos}, we see
    that $(\RR^\Xsf, \{T_\sigma\})$ is max-stable.
\end{example}

Regarding the definition of max-stability (resp., min-stability), the Bellman
min- and max-operators are often contraction maps and existence of a fixed point
is easily verified (see, e.g., \cite{denardo1967contraction} or Chapter~2 of
\cite{bertsekas2022abstract}). Here is another useful condition, which covers the case where state and action spaces are finite.

\begin{proposition}\label{p:fposet}
    Let $\aA$ be order stable and suppose the set of policies is finite.  In
    this setting,
    \begin{enumerate}
        \item if each $v \in V$ has at least one max-greedy policy, then $\aA$ is max-stable, and
        \item if each $v \in V$ has at least one min-greedy policy, then $\aA$
            is min-stable.
    \end{enumerate}
\end{proposition}

Proposition~\ref{p:fposet} is proved in the appendix (page~\pageref{pp:fposet}).

\begin{example}
    \label{eg:risksens_ms}
    The risk-sensitive MDP from Example~\ref{eg:risksens} is max-stable.
    Evidently $T_\sigma^\theta$ is order-preserving on $V=\RR^\Xsf$.
    Moreover, $(V, T_\sigma^\theta)$ is globally stable (see, e.g.,
    \cite{bauerle2018stochastic}) and hence order stable, by
    Example~\ref{eg:pspace}. This shows that $(V, \{T_\sigma^\theta\})$ is an
    order stable ADP. Given $v \in V$, we can construct a $v$-max-greedy
    policy $\sigma$ by setting
    \begin{equation*}
        \sigma(x) \in \argmax
        \left\{
            r(x, a) + \frac{\beta}{\theta}
            \ln \left[ 
                \sum_{x'} \exp(\theta v(x')) P(x, a, x')
            \right]
        \right\}
    \end{equation*}
    for all $x \in \Xsf$.  As the policy set $\Sigma$ is finite (since $\Xsf$
    and the choice sets are all finite), Proposition~\ref{p:fposet} implies that
    $(V, \{T_\sigma^\theta\})$ is max-stable.
\end{example}

\section{Optimality}\label{s:aop}

In this section we state conditions for optimality in our abstract setting.
These conditions are closely related to those in \cite{sargent2025partially}.
Later, in Section~\ref{s:iso}, we will study how optimality properties are
preserved under transformations.

\subsection{Max-Optimality}\label{ss:maop}

Let $\aA$ be a well-posed ADP. We let $V_\Sigma := \{v_\sigma\}_{\sigma \in
\Sigma}$ denote the set of $\sigma$-value functions.
If $V_\Sigma$ has a greatest element, then we denote it by $\vmax$ and call it
the \navy{max-value function}. A policy $\sigma \in \Sigma$ is called
\navy{max-optimal} for $\aA$ if $v_\sigma$ is a greatest element of $V_\Sigma$;
that is, if $\vmax$ exists and $v_\sigma = \vmax$.

We define a map $\Hmax$ from $V$ to $\{v_\sigma\}$ via
    $\Hmax \, v = v_\sigma$ where $\sigma$ is $v$-max-greedy.
Iterating with $\Hmax$ is an abstraction of Howard policy
iteration.\footnote{For $\Hmax$ to be well-defined, we must always select
    the same $v$-greedy policy when the operator is applied to $v$. We can use the
    axiom of choice to assign to each $v$ a designated $v$-greedy policy,
    although, in applications, a simple rule usually suffices.
    For example, if $\Sigma$ is finite, we can
    enumerate the policy set $\Sigma$ and choose the
    first $v$-greedy policy.} In what follows, we call $\Hmax$ the \navy{Howard max-operator}
generated by the ADP.

In the following result, we take $\aA$ to be an ADP with Bellman operator
$\tmax$.

\begin{theorem}[Max-optimality]\label{t:fbk}
    If $\aA$ is max-stable, then 
    \begin{enumerate}
        \item the max-value function $\vmax$ exists in $V$, 
        \item $\vmax$ is the unique solution to the Bellman max-equation in $V$, 
        \item a policy is max-optimal if and only if it is $\vmax$-max-greedy.
        \item at least one max-optimal policy exists.
    \end{enumerate}
    If, in addition, $\Sigma$ is finite, then Howard max-policy iteration
    converges to $\vmax$ in finitely many steps.
\end{theorem}

The last statement means that, for all $v \in V$, there exists a $K \in \NN$
such that $k \geq K$ implies $\Hmax^k v = \vmax$. The proof of
Theorem~\ref{t:fbk} is given in the appendix.

Analogous results hold for minimization; see Appendix~\ref{app:minim} for
the definitions of min-optimality, the Howard min-operator, dual ADPs, and
the corresponding min-optimality theorem.

\section{Isomorphic ADPs}\label{s:iso}

In this section we introduce isomorphic relationships between ADPs and explore
their implications for optimality.  True to their  name, isomorphic
relationships are symmetric, transitive and reflexive.
We show that isomorphic ADPs have identical optimality properties.

%If $V$ and $\hat V$ are lattices, $F$ is an order isomorphism, and $\{v_i\}_{i \in
%I}$ is a finite subset of $V$, then $F \bigvee_i v_i = \bigvee_i F v_i$ and $F
%\bigwedge_i v_i = \bigwedge_i F v_i$. If
%$F$ is an anti-isomorphism, then $F \bigvee_i v_i = \bigwedge_i F v_i$ and $F
%\bigwedge_i v_i = \bigvee_i F v_i$.

\subsection{Definition and Properties}

Let $\aA = (V, \{T_\sigma\})$ and $\hat \aA = (\hat V, \{\hat T_\sigma\})$ be
two ADPs. We call $\aA$ and $\hat \aA$ \navy{isomorphic}\index{Isomorphic
ADPs} under $F$ if these two ADPs have the same policy set $\Sigma$ and
$F$ is an order isomorphism from $V$ to $\hat V$ such
that 
\begin{equation}\label{eq:adpih}
    F  \circ T_\sigma = \hat T_\sigma  \circ F 
    \;\;
    \text{ on } V
    \text{ for all } 
    \sigma \in \Sigma.
\end{equation}
In other words, $(V, T_\sigma)$ and $(\hat V, \hat T_\sigma)$ are order
conjugate under $F$ for all $\sigma \in \Sigma$.\footnote{While the definition
    requires that the two ADPs have the same policy set $\Sigma$, it
    suffices that the policy sets can be
put in one-to-one correspondence with each other.}

An illustration of isomorphic ADPs built from $Q$-factor formulations is
given in Appendix~\ref{app:qfac}. We also provide a further illustration
involving risk-sensitive and Kreps--Porteus preferences in
Section~\ref{ss:rskp}.

\begin{remark}
    For ADPs $\aA, \hat \aA$, let $\aA \sim \hat \aA$ indicate that $\aA$ and $\hat
    \aA$ are isomorphic.  It is elementary to show that the relation $\sim$ is
    reflexive, symmetric and transitive. Hence $\sim$ is an equivalence relation on
    the set of all ADPs.
\end{remark}

\subsection{Isomorphisms and Optimality}\label{sss:iaoth}

We seek a connection between  value functions and optimality  properties of
isomorphic ADPs. The theory below provides this relationship.
For all of this section, we take $\aA = (V, \{T_\sigma\})$ and
$\hat \aA = (\hat V, \{\hat T_\sigma\})$ to be two ADPs with the same policy
set. When they exist, we let 
\begin{itemize}
    \item $v_\sigma$ (resp., $\hat v_\sigma$) be the unique fixed point of $T_\sigma$
        (resp., $\hat T_\sigma$)
    \item $\tmax$ (resp., $\htmax$) be the Bellman max-operator of $\aA$ 
        (resp., $\hat \aA$) 
    \item $\tmin$ (resp., $\htmin$) be the Bellman min-operator of $\aA$ 
        (resp., $\hat \aA$) 
    \item $\vmax$ (resp., $\hvmax$) be the max-value function of $\aA$ 
        (resp., $\hat \aA$)
    \item $\vmin$ (resp., $\hvmin$) be the min-value function of $\aA$ 
        (resp., $\hat \aA$)
\end{itemize}

Isomorphic ADPs share the same regularity properties:

\begin{theorem}\label{t:iso}
    If $\aA$ and $\hat \aA$ are isomorphic under $F$, then 
    \begin{enumerate}
        \item $\aA$ is well-posed if and only if $\hat \aA$ is well-posed.
        \item $\aA$ is order stable if and only if $\hat \aA$ is order stable.
        \item $\aA$ is max-stable if and only if $\hat \aA$ is max-stable.
            In this case, 
            \begin{equation}\label{eq:isotmax}
                F \circ \tmax = \htmax \circ F
                \quad \text{and} \quad
                \hvmax = F \, \vmax.  
            \end{equation}
            Moreover, $\aA$ and $\hat \aA$ have the same 
             max-optimal policies.
        \item $\aA$ is min-stable if and only if $\hat \aA$ is min-stable.
            In this case, 
            \begin{equation}\label{eq:isotmin}
                F \circ \tmin = \htmin \circ F
                \quad \text{and} \quad
                \hvmin = F \, \vmin.  
            \end{equation}
            Moreover, $\aA$ and $\hat \aA$ have the same min-optimal policies.
    \end{enumerate}
\end{theorem}

The proof (page~\pageref{pf:iso}) uses condition \eqref{eq:adpih} to show
that $F$ commutes with the Bellman operators, and then combines
Theorem~\ref{t:fbk} with properties of order conjugate systems.

\subsection{Anti-Isomorphic ADPs}\label{ss:aiadps}

Let $\aA = (V, \{T_\sigma\})$ and $\hat \aA = (\hat V, \{\hat T_\sigma\})$ be
two ADPs. We call $\aA$ and $\hat \aA$ \navy{anti-isomorphic} under $F$ if
they have the same policy set $\Sigma$ and, in addition, there exists an
order anti-isomorphism $F$ from $V$ to $\hat V$ such that \eqref{eq:adpih}
holds (see Example~\ref{eg:fei}).

If $\aA \asim \hat \aA$ indicates that $\aA$ and $\hat \aA$ are anti-isomorphic,
then $\asim$ is symmetric and transitive but, in general, not reflexive.

Here is an optimality result for anti-isomorphic ADPs that parallels
Theorem~\ref{t:iso}.

\begin{theorem}\label{t:antiiso}
    If $\aA$ and $\hat \aA$ are anti-isomorphic under $F$, then 
    \begin{enumerate}
        \item $\aA$ is well-posed if and only if $\hat \aA$ is well-posed.
        \item $\aA$ is order stable if and only if $\hat \aA$ is order stable.
        \item $\aA$ is max-stable if and only if $\hat \aA$ is min-stable.
            In this case, 
            \begin{equation}\label{eq:antiisotmax}
                F \circ \tmax = \htmin \circ F
                \quad \text{and} \quad
                \hvmin = F \, \vmax.  
            \end{equation}
            Moreover, $\sigma \in \Sigma$ is max-optimal for $\aA$ if and only
            if $\sigma$ is min-optimal for $\hat \aA$.
    \end{enumerate}
\end{theorem}

The proof of Theorem~\ref{t:antiiso} is given in the appendix
(page~\pageref{pf:antiiso}).

\section{Applications}\label{s:app}

In this section we show how isomorphic relationships can
simplify or illuminate dynamic programming problems.

\subsection{Modified Epstein--Zin Equations}\label{ss:modep}

Consider an Epstein--Zin version of the MDP in Example~\ref{eg:mdp}  (see, e.g.,
\cite{epstein1989risk} or \cite{weil1990nonexpected}), in which a Bellman
max-equation takes the form
\begin{equation*}
    v(x) =
    \max_{a \in \Gamma(x)}
        \left\{
            r(x, a)^\alpha + \beta(x)
            \left( 
                \sum_{x'} v(x')^\gamma P(x, a, x')
            \right)^{\alpha/\gamma}
        \right\}^{1/\alpha}.
\end{equation*}
Following much of the recent literature, we have allowed the rate of time preference
$\beta$ to depend on the state (see, e.g., \cite{albuquerque2016valuation},
\cite{schorfheide2018identifying}, \cite{degroot2018},
\cite{gomez2020important}). The function $\beta$ maps $\Xsf$ to $\RR_+$ and
$\gamma$ and $\alpha$ are nonzero parameters. Other details are as in
Example~\ref{eg:mdp}. Using the symbols $r_\sigma$ and $P_\sigma$ from
\eqref{eq:rps}, the policy operator $T_\sigma$ can be written as
\begin{equation}\label{eq:ezfp2}
    T_\sigma \, v
    = \left\{
        r_\sigma^\alpha + \beta 
            \left( P_\sigma \, v^\gamma \right)^{\alpha / \gamma}
    \right\}^{1/\alpha},
\end{equation}
where powers are taken pointwise.
Since $\gamma$ and $\alpha$ can be negative, we assume that $r$ is positive. We
also suppose that $P_\sigma$ is irreducible for all $\sigma \in
\Sigma$. Under these assumptions, $T_\sigma$ is an order-preserving self-map on
$(V, \leq)$, the set of all strictly positive functions on $\Xsf$
paired with the pointwise partial order, and $\aA := (V, \{T_\sigma\})$ is an
ADP.

Now set
\begin{equation}\label{eq:bezt}
    \theta := \frac{\gamma}{\alpha}
    \quad \text{and} \quad
    \hat T_\sigma \, v
    := \left\{
            r_\sigma^\alpha
            + \beta \left( P_\sigma v \right)^{1/\theta}
        \right\}^{\theta}.
\end{equation}
The pair $\hat \aA = (V, \{\hat T_\sigma\})$ is also an ADP.

\begin{lemma}\label{l:f}
    The following relationships hold:
    \begin{enumerate}
        \item If $\gamma > 0$, then $\aA$ and $\hat \aA$ are isomorphic.  
        \item If $\gamma < 0$, then $\aA$ and $\hat \aA$ are anti-isomorphic.  
    \end{enumerate}
\end{lemma}

\begin{proof}
    Let $F \colon V \to V$ be the bijective map $Fv = v^\gamma$. 
    Fixing $\sigma$ and applying \eqref{eq:ezfp2} yields
    \begin{equation*}
       F \, T_\sigma \, v
        = (T_\sigma \, v)^\gamma 
        = \left\{
            r_\sigma^\alpha + \beta 
                \left( P_\sigma \, v^\gamma \right)^{\alpha / \gamma}
        \right\}^{\gamma/\alpha}
        = \left\{
            r_\sigma^\alpha + \beta 
                \left( P_\sigma \, v^\gamma \right)^{1/ \theta}
        \right\}^{\theta}.
    \end{equation*}
    Inspection of \eqref{eq:bezt} shows that $\hat T_\sigma \, Fv = \hat
    T_\sigma \, v^\gamma$ is identical
    to the last expression in the display above.
    Hence $F \circ T_\sigma = \hat T_\sigma \circ F$ on $V$.  If $\gamma > 0$, then $F$ is
    order-preserving, so $\aA$ and $\hat \aA$ are isomorphic. If $\gamma < 0$,
    then $F$ is order-reversing, so $\aA$ and $\hat \aA$ are anti-isomorphic.
\end{proof}

Lemma~\ref{l:f} gives us a way to solve  for max-optimal policies of $\aA$ by
studying $\hat \aA$ (and applying either Theorem~\ref{t:iso} or
Theorem~\ref{t:antiiso}).  This is convenient because $\hat \aA$ is easier to
analyze.  The next section illustrates.

\subsection{Characterizing Optimality}\label{ss:ezexact}

Suppose that $\beta$ depends on $x$ 
through a purely exogenous state component (as in, say, \cite{degroot2018} and
\cite{schorfheide2018identifying}).  Specifically, $\Xsf = \Ysf \times
\Zsf$ and $x = (y, z)$, where
\begin{equation*}
    P(x, a, x') = R(y, a, y') Q(z, z')
    \quad \text{and} \quad \beta(x) = \beta(z).
\end{equation*}
Here $R(y, a, \cdot)$ is a distribution over $y$ for each feasible $(y, a)$ pair
and $Q$ is a stochastic matrix over $\Zsf$.  For each $z \in \Zsf$, let
$(Z_t(z))_{\, t \geq 0}$ be a Markov chain on $\Zsf$ generated by $Q$ and
starting at $z$. Define
\begin{equation}\label{eq:elldef}
    \eE(\beta, Q, \theta) 
    := \lim_{k \to \infty}
        \left\{
            \sup_{z \in \Zsf}
            \EE \prod_{t=0}^{k-1} \beta(Z_t(z))^\theta
        \right\}^{1/k}.
\end{equation}

We can now state the following exact result.

\begin{theorem}\label{t:eziff}
    The Epstein--Zin ADP $\aA$ is max-stable if and only if
    \begin{equation}\label{eq:ezkk}
        \eE(\beta, Q, \theta)^{1/\theta} < 1
    \end{equation}
    Hence, under \eqref{eq:ezkk}, all of the optimality results in
    Theorem~\ref{t:fbk} apply. Conversely, if \eqref{eq:ezkk} fails, then $\aA$
    is not well-posed and optimality is undefined.
\end{theorem}

The sufficiency of \eqref{eq:ezkk} for optimality properties is related to an
earlier result in \cite{stachurski2021dynamic}. It is also connected to
\cite{bloise2024not}, who establish sufficient conditions for optimality via
the spectral radius of a monotone sublinear operator. The converse result is
new. Unlike the existing literature, the proof of Theorem~\ref{t:eziff}
proceeds by studying the simpler ADP $\hat \aA$ introduced in
Section~\ref{ss:modep}, which we show is either isomorphic or
anti-isomorphic, depending on the sign of $\gamma$. Then the theorem follows
from Proposition~\ref{p:fposet} together with a result on irreducible
positive operators. See Section~\ref{ss:pezr} of the appendix.

\subsection{Kreps--Porteus vs Risk Sensitive Preferences}\label{ss:rskp}

As a further application, we show that two widely used preference
specifications, the risk-sensitive preferences and a multiplicative variation
of Kreps--Porteus preferences, are connected by an isomorphism. This allows
optimality results from the well-studied risk-sensitive framework to be
transferred to the multiplicative Kreps--Porteus model.

We continue with the finite MDP framework of Example~\ref{eg:mdp}, using
$r_\sigma$ and $P_\sigma$ as in \eqref{eq:rps}. Recall the risk-sensitive
MDP in Example~\ref{eg:risksens}, where the policy operators take the form
\begin{equation}\label{eq:brs}
    (T_\sigma^\theta \, v)(x)
    = r_\sigma(x) + \frac{\beta}{\theta}
        \ln \left[
            \sum_{x'} \exp(\theta v(x')) P(x, \sigma(x), x')
        \right]
\end{equation}
for $v \in \RR^\Xsf$ and $\theta \in \RR \setminus \{0\}$. We show in
Example~\ref{eg:risksens_ms} that the ADP $(\RR^\Xsf, \{T_\sigma^\theta\})$
is max-stable and hence all of the optimality results in Theorem~\ref{t:fbk}
apply.

An alternative formulation is to replace the entropic certainty equivalent
with Kreps--Porteus expectations, leading to the policy operator
\begin{equation}\label{eq:kp}
    (K_\sigma \, v)(x)
    = r_\sigma(x) + \beta
    \left\{
        (P_\sigma \, v^\nu)(x)
    \right\}^{1/\nu}
\end{equation}
for $v \in (0, \infty)^\Xsf$ and $\nu \in \RR \setminus \{0\}$, where $r$ is
strictly positive. The Kreps--Porteus ADP $((0, \infty)^\Xsf, \{K_\sigma\})$
is harder to analyze directly because $K_\sigma$ is not generally a
contraction. There is, however, a multiplicative variation on the
Kreps--Porteus ADP that is simple to analyze. The model is obtained by
setting
\begin{equation}\label{eq:mkp}
    (M_\sigma \, v)(x)
    = r_\sigma(x) \cdot
    \left\{
        (P_\sigma \, v^\nu)(x)
    \right\}^{\beta / \nu}
\end{equation}
where $v \in (0, \infty)^\Xsf$, $r$ is strictly positive, and $\beta \in
[0, 1)$.  We call $\hat \aA := ((0, \infty)^\Xsf, \{M_\sigma\})$ the
multiplicative Kreps--Porteus (MKP) ADP.

It turns out that the MKP ADP and the risk-sensitive ADP are isomorphic. To
see this, we take logs of the Bellman max-equation associated with the MKP
ADP, obtaining
\begin{equation*}
    \ln v(x) =
    \max_{a \in \Gamma(x)}
    \left\{
        \ln r(x,a) + \frac{\beta}{\nu} \ln
        \left[
            \sum_{x'} v(x')^\nu P(x, a, x')
        \right]
    \right\}.
\end{equation*}
Setting $\hat v = \ln v$ and $\hat r = \ln r$ yields
\begin{equation*}
    \hat v(x) =
    \max_{a \in \Gamma(x)}
    \left\{
        \hat r(x,a) + \frac{\beta}{\nu} \ln
        \left[
            \sum_{x'} \exp(\nu \, \hat v(x')) P(x, a, x')
        \right]
    \right\}.
\end{equation*}
This is exactly the Bellman equation for \eqref{eq:brs}, after replacing $r$
with $\hat r$ and $\theta$ with $\nu$.

We can formalize this observation.  The MKP ADP is $\hat \aA = ((0,
\infty)^\Xsf, \{M_\sigma\})$.  Let $F \colon (0, \infty)^\Xsf \to
\RR^\Xsf$ be the pointwise logarithm, $Fv = \ln v$, which is an order
isomorphism with inverse $F^{-1} = \exp$.  Consider the risk-sensitive ADP
$\aA := (\RR^\Xsf, \{T_\sigma^\nu\})$ with reward $\hat r = \ln r$ and
risk-sensitivity parameter $\theta = \nu$.

It remains to verify condition \eqref{eq:adpih}, i.e., that $M_\sigma =
F^{-1} \circ T_\sigma^\nu \circ F$ on $(0, \infty)^\Xsf$. Setting $\hat v =
Fv = \ln v$ for $v \in (0, \infty)^\Xsf$, we have
\begin{align*}
    (F^{-1} \circ T_\sigma^\nu \circ F)(v)(x)
    &= \exp\!\left[
        \ln r_\sigma(x) + \frac{\beta}{\nu}
        \ln \left[
            \sum_{x'} \exp(\nu \, \hat v(x')) P(x, \sigma(x), x')
        \right]
    \right] \\
    &= r_\sigma(x) \cdot
    \left\{
        \sum_{x'} \exp(\nu \, \ln v(x')) P(x, \sigma(x), x')
    \right\}^{\beta/\nu} \\
    &= r_\sigma(x) \cdot
    \left\{ (P_\sigma \, v^\nu)(x) \right\}^{\beta/\nu}
    = (M_\sigma \, v)(x).
\end{align*}
Since $F$ is an order isomorphism, the MKP ADP $\hat \aA$ and the
risk-sensitive ADP $\aA$ are isomorphic.  By Theorem~\ref{t:iso}, they share
all optimality properties: $\hat \aA$ is max-stable if and only if $\aA$ is
max-stable, the value functions are related by $\hvmax = F \vmax$, and the two
ADPs have the same optimal policies.

This result provides a practical benefit.  Since the risk-sensitive MDP is
an additive model with well-developed theory
\citep{bauerle2018stochastic}, the isomorphism allows us to immediately
establish optimality results for the multiplicative Kreps--Porteus model
without the need for separate analysis.  The same isomorphism extends
naturally to continuous state space settings, where sums are replaced by
expectations and $\RR^\Xsf$ by bounded measurable functions.

\subsection{Curvature and Approximation}

Implementing dynamic programs with continuous state spaces on computers requires
approximation of value functions, policy functions, or both.  The need
for high quality approximations of these functions is amplified by the
fact that, in dynamic programming, most computational algorithms use some form
of iteration (such as value or Howard policy iteration), and approximation
errors tend to compound across iterations, potentially leading to poor final
approximations or failures of convergence (see, e.g., \cite{Farahmand2010Error}).

In general, functions are harder to approximate when they involve high degrees
of curvature.  Intuitively, isomorphic relationships can improve accuracy of
approximations and stabilize iteration
by transforming dynamic programs in order to reduce curvature or other
complexities in target functions.  We explore this idea in the current section,
using the multisector real business cycle model of \cite{long1983real} as a test case.

In their model, the Bellman max-equation has the form
\begin{equation}
    \label{eq:long_bellman}
    v(y) = \max_{c, X} 
    \left\{ 
        u(c) + \beta \int v(f(\lambda, X)) \phi(\diff \lambda)
    \right\},
\end{equation}
subject to 
\begin{equation}
    \label{eq:long_cons}
    0 \leq c_i, X_{ij}, y_i
    \quad \text{and} \quad
    c_j + \sum_{i=1}^n X_{ij} = y_j
    \quad \text{for all $i, j$ in $\{1, \ldots, n\}$}.
\end{equation}
The function $f$ is a Borel measurable production function taking values in $\RR^n_+$,
$c = (c_j)$ is an $n$-vector of consumption quantities across $n$ goods,
$y = (y_j)$ is an $n$-vector of final outputs, $X = (x_{ij})$ is an $n \times n$
matrix of commodity inputs, and $\phi$ is a distribution over $\RR_+^n$. 
(The labor-leisure decision is missing from \eqref{eq:long_bellman} because we
adopt a special case of the parameterization in \cite{long1983real} that allows
us to assume inelastic labor supply.)

Let the state space be $\Ysf = \RR^n_+$ and let $V = m\RR^\Ysf$, the space of
(extended) real-valued Borel measurable functions on $\Ysf$. The policy operator is
given by
\begin{equation*}
    (T_\sigma v)(y) = u(\sigma_c(y)) + \beta \int v\left[
        f(\lambda, \sigma_X(y))
    \right] \phi(d\lambda) \qquad (y \in \Ysf)
\end{equation*}
where $v \in V$ and $\sigma = (\sigma_c, \sigma_X)$ is a Borel measurable
feasible policy such that $(c, X) = \sigma(y)$ satisfies the constraints in
\eqref{eq:long_cons}. Then the model represented by the pair $(V,
\{T_\sigma\})$ is an ADP. In many settings, such as the one we will study
below, this program is well-posed and max-stable. In view of
Theorem~\ref{t:iso}, any isomorphic program shares the same optimality and
convergence properties of the original problem. Therefore, we can
equivalently solve the transformed Bellman equation
\begin{equation}\label{eq:long_bellman_c}
    w(y) = \max_{c, X} 
    F\left[
        u(c) + \beta \int (F^{-1}\circ w)(f(\lambda, X))
        \phi(\diff \lambda) 
    \right]
\end{equation}
for any strictly increasing function $F\colon \RR \to \RR$, and recover the
value function of the original problem via $v = F^{-1} w$. (We understand $F$
as both a strictly increasing function from $\RR$ to itself and as an order
isomorphism on $V$ sending $v$ into $y \mapsto F(v(y))$.) This is
advantageous because, for suitably chosen $F$, solving
\eqref{eq:long_bellman_c} significantly improves numerical accuracy.

To illustrate this point, we follow
\cite{long1983real} in assuming that utility obeys $U(c) = \sum_{i=1}^n \theta_i
\ln c_i$ and production has the multisector Cobb-Douglas form $f(y,
\lambda) = \lambda \prod_{j=1}^n x_{ij}^{a_{ij}}$. Under these assumptions, an
analytical solution to
\eqref{eq:long_bellman} exists, allowing us to compare approximate numerical
solutions with exact solutions.  The analytical solution has the form $v_0(y) = \sum_{i=1}^n \gamma_i\ln y_i +
K$, where $\gamma_i$ and $K$ depend on the parameters $a_{ij}$ and $\theta_i$. 
This function has relatively high curvature near zero, with both slope and
curvature approaching infinity as output converges to the origin. These features
make it challenging to approximate with a high degree of accuracy. In contrast, if we consider an isomorphic
transformation \eqref{eq:long_bellman_c} with $F(x) = e^{mx+b}$, then the
transformed value function is given by
\begin{equation}
    \label{eq:w}
    (F v_0)(y) = e^{mK+b} \prod_{i=1}^n y_i^{m\gamma_i}.
\end{equation}
This function has lower curvature and finite derivatives at zero. These
properties facilitate approximation and lead to improved accuracy.

To demonstrate, we solve \eqref{eq:long_bellman} and
\eqref{eq:long_bellman_c} for the two-sector case under several
parameterizations using value function iteration where the value functions
are approximated by linear interpolation on a common grid. To ensure a fair
comparison, we perform a fixed number of iterations ($N = 500$) for both the
original and transformed problems. The maximization operations in
\eqref{eq:long_bellman} and \eqref{eq:long_bellman_c} are carried out using
the Adam optimizer \citep{kingma2014adam} for efficient computation. In each
maximization operation, we run the optimizer for 2000 iterations following a
cosine decay learning rate schedule with a 500-step warmup. The initial,
peak, and final learning rates of the schedule are $1\times 10^{-6}$,
$1\times 10^{-4}$, and $1\times 10^{-7}$, respectively, chosen through trial
and error. We then compute the sup-norm errors for the two problems, $\|v -
v_0\|$ and $\|F^{-1} w - v_0\|$. Table~\ref{tab:long} shows the improvement
in accuracy from solving the transformed problem, measured by $\|v -
v_0\|/\|F^{-1}\circ w - v_0\|$.\footnote{For each parameterization, we choose
  the parameters of the exponential transformation $F(x) = e^{mx+b}$ so that
  $F v_0$ is approximately linear in at least one dimension. In particular,
  in the second column, $m = [0.35, 0.3, 0.3, 0.25, 0.25, 0.25]$ and $b =
  [75, 60, 65, 60, 60, 75]$; in the third column, $m = [0.9, 0.8, 0.7, 0.7,
  0.6, 0.5]$ and $b = [90, 85, 80, 85, 80, 75]$. In all cases, the state
  space is $[1\times 10^{-7}, 20]^2$ with 500 grid points along each
  dimension. The other parameter values are listed in the table.}

\begin{table}[tb!]
    \caption{Accuracy Gain from Solving the Transformed Problem}
    \label{tab:long}
    \footnotesize
    \begin{tabular*}{\textwidth}{@{} l @{\extracolsep\fill} c @{\extracolsep\fill}  c @{}}
      \toprule
      Parameterization & $A = ([0.2, 0.7], [0.6, 0.1])$ & $A = ([0.1, 0.7],
                                                          [0.3, 0.1])$\\ 
      \cmidrule{2-3}
       & $\theta_2 = 1.0$ & $\theta_2 = 1.0$ \\
      \midrule
      $\theta_1 = 0.5$ & 104.6 & 68.8 \\
      $\theta_1 = 0.6$ & 153.6 & 65.3 \\
      $\theta_1 = 0.7$ & 143.6 & 63.1 \\
      $\theta_1 = 0.8$ & 109.7 & 62.6 \\
      $\theta_1 = 0.9$ & 173.3 & 178.9 \\
      $\theta_1 = 1.0$ & 178.7 & 122.3 \\
      \bottomrule
    \end{tabular*}
\end{table}

Although the transformed value function $w$ is not exactly linear for the
two-sector case (see \eqref{eq:w}), we can make it approximately linear in at
least one dimension by choosing $F$ appropriately (e.g., setting $m =
1/\gamma_1$). Table~\ref{tab:long} shows that this improved approximation of
the value function leads to solutions that are, on average, 100 times more
accurate across a wide range of parameterizations.

\section{Conclusion}\label{sec:concl}

We studied isomorphic and anti-isomorphic relationships between dynamic
programs and showed how optimality properties transmit across these
relationships. We applied these ideas to Epstein--Zin preferences with time
preference shocks, to the connection between multiplicative Kreps--Porteus
and risk-sensitive preferences, and to improving numerical accuracy in
multisector real business cycle models.

Our research focused on discrete time dynamic programs, which we paired with
discrete time concepts of topological and order conjugacy.
Continuous time dynamic programs are also important and the concept of
topological conjugacy has a well-defined analogy in continuous time dynamics.
This suggests that many of our ideas can be carried over to continuous time
systems.  We leave this task for future research.

\appendix

\section{Order-Theoretic Preliminaries}\label{app:order}

This appendix collects the order-theoretic definitions, dynamical systems
background, and auxiliary results that underpin the main text.

\subsection{Posets}\label{app:posets}

Let $V = (V, \preceq)$ be a partially ordered set, also called a poset. When $A
\subset V$, the symbol $\bigvee A$ refers to the supremum of $A$ in $V$ and
$\bigwedge A$ is the infimum  (see, e.g., \cite{davey2002introduction}).
$V$ is called \navy{bounded} if $V$ has a least and greatest element. A
sequence $(v_n)$ in $V = (V, \preceq)$ is called \navy{increasing} if $v_n
\preceq v_{n+1}$ for all $n \in \NN$.  If $(v_n)$ is increasing and $\bigvee_n
v_n = v$ for some $v \in V$, then we write $v_n \uparrow v$.  The set $V$ is
called \navy{countably chain complete} if $V$ is bounded and every increasing
sequence in $V$ has a supremum in $V$.   A self-map $S$ on $V$ is called
\navy{order continuous} on $V$ if $Sv_n \uparrow Sv$ whenever $v_n \uparrow v$.

\subsection{Dynamical Systems}\label{app:ds}

A \navy{dynamical system} is a pair $(V, S)$ where $V$ is a set and $S$ is a
self-map on $V$.  Given $v \in V$, the sequence $(S^n v)_{n \geq 1}$ is called
the \navy{trajectory} of $v$ under $S$. When $V$ has a topology, we can study
the convergence of trajectories.  In particular, in this setting, a system $(V,
S)$ is called \navy{globally stable} when $S$ has a unique fixed point $\bar v$
in $V$ and $\lim_{k \to \infty} S^k v = \bar v$ for all $v \in V$.

Two dynamical systems $(V, S)$ and $(\hat V, \hat S)$ are called
\navy{conjugate} under $F$ if there exists a bijection $F$ from $V$ to $\hat V$
such that $F \circ S = \hat S \circ  F$ on $V$.  In this setting, if $v$ is a
fixed point of $S$, then $\hat S F v = F S v = F v$, so $Fv$ is a fixed point of
$\hat S$.  With slightly more effort, one can show the following.

\begin{lemma}\label{l:foo000}
    If $(V, S)$ and $(\hat V, \hat S)$ are conjugate under $F$ and
    $S$ has a unique fixed point $v \in V$, then $Fv$ is the unique fixed
    point of $\hat S$ in $\hat V$.
\end{lemma}

For more details see, for example, \cite{sternberg2014dynamical} or \cite{layek2015introduction}.

Dynamical systems $(V, S)$ and $(\hat V, \hat S)$ are called \navy{topologically
conjugate} under $F$ if both $V$ and $\hat V$ have topologies, the pairs $(V,
S)$ and $(\hat V, \hat S)$ are conjugate under $F$ and, in addition, $F$ is a
homeomorphism (i.e., continuous bijection with continuous inverse) from $V$ to
$\hat V$. In this setting we can extend Lemma~\ref{l:foo000} as follows.

\begin{lemma}\label{l:foo00}
    If $(V, S)$ and $(\hat V, \hat S)$ are topologically conjugate under $F$,
    then $(V, S)$ is globally stable if and only if $(\hat V, \hat S)$ is
    globally stable.
\end{lemma}

Lemma~\ref{l:foo00} is standard and routinely used in dynamical systems theory.
(There are similar results concerning local stability for topologically
conjugate systems, as well as additional results concerning other kinds of
asymptotic behavior.  See, for example, \cite{sternberg2014dynamical} or
\cite{layek2015introduction}.)

\subsection{Order Stability}\label{app:os}

Let $(V, S)$ be a dynamical system where $V$ is partially ordered by $\preceq$
and $S$ has a unique fixed point $\bar v$ in $V$.  In this environment,
\cite{sargent2025partially} call $(V, S)$
\begin{enumerate}
    \item \navy{upward stable} if $v \in V$ and $v \preceq S \, v$  implies $v
        \preceq \bar v$,
    \item \navy{downward stable} if $v \in V$ and $S \, v \preceq v$ implies $\bar
        v \preceq v$, and
    \item \navy{order stable} if upward and downward stability both hold.
\end{enumerate}

In this paper, we discuss both maximization and minimization.  To link these
ideas in a poset environment, we use the notion of order duals.  In particular,
given poset $V = (V, \preceq)$, let $(V, \preceq^\partial)$ be the
order dual, so that, for $u, v \in V$, we have $u \preceq^\partial v$ if and
only if $v \preceq u$. For convenience, we sometimes
denote $(V, \preceq^\partial)$ by $V^\partial$.

\begin{lemma}\label{l:odod}
    $(V, S)$ is order stable if and only if $(V^\partial, S)$ is order stable.
\end{lemma}

\begin{proof}[Proof of Lemma~\ref{l:odod}]\label{pf:odod}
    Let $(V, S)$ be order stable.  By definition, $S$ has a unique fixed point
    $\bar v \in V$. We claim that $(V^\partial, S)$ is upward
    and downward stable. Regarding upward stability, suppose $v \in V$ and $v
    \preceq^\partial S v$.  Then $Sv \preceq v$ and hence $\bar v \preceq v$, by
    downward stability of $(V, S)$.  But then $v \preceq^\partial \bar v$, so
    $(V^\partial, S)$ is upward stable.  The proof of downward stability is
    similar.  Hence order stability of $(V, S)$ is sufficient for order
    stability of $(V^\partial, S)$. Necessity follows from sufficiency, since
    the dual of $(V^\partial, S)$ is $(V, S)$.
\end{proof}

\subsection{Order Conjugacy}\label{app:oc}

Our next step is to replace the concept of topological conjugacy from
Section~\ref{app:ds} with a parallel notion for dynamical systems on posets.
To this end, we recall that a bijective map $F \colon V \to \hat V$ is called
an \navy{order isomorphism} if both $F$ and its inverse $F^{-1}$ are order
preserving. (We call $F \colon V \to \hat V$ an \navy{order anti-isomorphism}
if both $F$ and $F^{-1}$ are order reversing.) Let $(V, S)$ and $(\hat V,
\hat S)$ be two dynamical systems where $V$ and $\hat V$ are partially
ordered. We call $(V, S)$ and $(\hat V, \hat S)$ \navy{order conjugate} under
$F$ when $(V, S)$ and $(\hat V, \hat S)$ are conjugate under $F$ and, in
addition, $F$ is an order isomorphism. It is easy to verify that order
conjugacy is an equivalence relation on the set of dynamical systems over
partially ordered sets.

\begin{lemma}\label{l:ocos}
    If $(V, S)$ and $(\hat V, \hat S)$ are order conjugate under $F$,
    then $(V, S)$ is order stable if and only if $(\hat V, \hat S)$ is order stable.
\end{lemma}

\begin{proof}[Proof of Lemma~\ref{l:ocos}]\label{pf:ocos}
    Let $(V, S)$ and $(\hat V, \hat S)$ be order conjugate under $F$,
    and suppose that $(V, S)$ is order stable.  By
    Lemma~\ref{l:foo000}, the map $\hat S$ has a unique fixed point in $\hat V$.
    Let $\hat w$ be an element of $\hat V$
    satisfying $\hat S\hat w \preceq \hat w$. Let $v$ and $\hat v := Fv$ be
    the fixed points of $S$ and $\hat S$, respectively. Then $F^{-1}
    \hat S \hat w \preceq F^{-1} \hat w$ and hence $S F^{-1} \hat w \preceq
    F^{-1} \hat w$.  But then $v \preceq F^{-1} \hat w$, by downward stability
    of $(V, S)$. Applying $F$ gives $\hat v \preceq \hat w$. Hence $(\hat V,
    \hat S)$ is downward stable. Similarly, if $\hat w$ is an element of $\hat
    V$ satisfying $\hat w \preceq \hat S\hat w$, then $F^{-1} \hat w \preceq
    F^{-1} \hat S \hat w = S F^{-1} \hat w$.  By upward stability of $(V, S)$,
    we have $F^{-1} \hat w \preceq v$. Applying $F$ gives $\hat w \preceq \hat
    v$, so $(\hat V, \hat S)$ is upward stable. Together, these results show
    that $(\hat V, \hat S)$ is order stable.
    The converse implication follows from symmetry.
\end{proof}

\section{Additional Examples}\label{app:qfac}

The following examples extend the ADP framework of Section~\ref{s:adps} to
$Q$-factor formulations from the reinforcement learning literature. All
primitives are as in Example~\ref{eg:mdp}.

\begin{example}[$Q$-factors]\label{eg:qfac}
    The $Q$-learning literature studies the $Q$-factor Bellman equation (see, e.g.,
    \cite{kochenderfer2022algorithms}, Section~17.2), which is given by
    \begin{equation}\label{eq:qbell}
        f(x, a) = r(x, a) +
        \beta \sum_{x'} \max_{a' \in \Gamma(x')} f(x', a')
            P(x, a, x')
            \qquad ((x,a) \in \Gsf).
    \end{equation}
    Here $f \in \RR^\Gsf$ is called the \navy{$Q$-factor}. The
    policy operators over $Q$-factors take the form
    \begin{equation}\label{eq:pabop}
        (Q_\sigma \, f)(x, a) = r(x, a) +
            \beta \sum_{x'} f(x', \sigma(x'))
            P(x, a, x'),
    \end{equation}
    where $f \in \RR^\Gsf$ and $\sigma \in \Sigma$.
    If we pair $\RR^\Gsf$ (the set of all real-valued functions on $\Gsf$)
    with the pointwise partial order $\leq$ and $Q_\sigma$ as in
    \eqref{eq:pabop}, then $(\RR^\Gsf, \{Q_\sigma\})$ is an ADP.
\end{example}

\begin{example}[Risk-sensitive $Q$-factors]\label{eg:rsqfac}
    It has become popular in reinforcement learning and related fields
    to extend the $Q$-factor approach from Example~\ref{eg:qfac} to
    risk-sensitive decision processes (see, e.g., \cite{fei2021exponential}). The
    corresponding $Q$-factor Bellman equation is given by
    \begin{equation}\label{eq:rsqbell}
        f(x, a) =
            r(x, a) + \frac{\beta}{\theta}
            \ln
            \left\{
                    \sum_{x'} \exp \left[ \theta \max_{a' \in \Gamma(x')} f(x', a')
                \right] P(x, a, x')
            \right\}
            \qquad ((x,a) \in \Gsf).
    \end{equation}
    The policy operators over risk-sensitive $Q$-factors take the form
    \begin{equation}\label{eq:rspbell}
        (Q_\sigma^\theta \, f)(x, a) =
            r(x, a) + \frac{\beta}{\theta}
            \ln \left[
                    \sum_{x'} \exp \left[ \theta f(x', \sigma(x'))
                \right] P(x, a, x')
            \right]
    \end{equation}
    where $f \in \RR^\Gsf$ and $\sigma \in \Sigma$.   The pair $(\RR^\Gsf,
    \{Q_\sigma^\theta\})$ is an ADP.
\end{example}

\begin{example}[Exponential risk-sensitive $Q$-factors]\label{eg:fei}
    \cite{fei2021exponential} work with an ``exponential'' risk-sensitive
    $Q$-factor Bellman equation where the corresponding policy operators have
    the form
    \begin{equation*}
        (M_\sigma \, h)(x, a)
        = \exp
        \left\{
            \theta r(x, a) + \beta \ln
            \left[
                \sum_{x'} h(x', \sigma(x')) P(x, a, x')
            \right]
        \right\}
    \end{equation*}
    Here $h \in (0, \infty)^\Gsf$ and $M_\sigma$ maps $(0, \infty)^\Gsf$ into itself.
    All primitives are as in the risk-sensitive $Q$-factor ADP $\aA :=
    (\RR^\Gsf, \{Q_\sigma\})$ in Example~\ref{eg:rsqfac}. Since each $M_\sigma$
    is order preserving (under the usual pointwise order), the pair $\hat \aA :=
    ((0, \infty)^\Gsf, \{M_\sigma\})$ is also an ADP.
    If we take $F$ to be the bijection from $\RR^\Gsf$ to $(0, \infty)^\Gsf$ defined
    by $(Fh)(x, a) = \exp(\theta h(x,a))$, then, for $Q_\sigma$ defined in
    \eqref{eq:rspbell}, $h \in \RR^\Gsf$ and $(x, a) \in \Gsf$,
    \begin{equation*}
        (F Q_\sigma \, h)(x, a)
        = \exp
        \left\{
            \theta r(x, a) + \beta
            \ln \left[
                    \sum_{x'} \exp \left[ \theta h(x', \sigma(x'))
                \right] P(x, a, x')
            \right]
        \right\}.
    \end{equation*}
    This is equal to $(M_\sigma F h)(x, a)$, which shows that $F \circ Q_\sigma = M_\sigma
    \circ F$ on $\RR^\Gsf$.  As a consequence, $\aA$ and $\hat \aA$ are isomorphic
    when $\theta > 0$ and anti-isomorphic when $\theta < 0$.
\end{example}

\section{Remaining Proofs}

\subsection{Proof of Optimality Results}

Let $\aA = (V, \{T_\sigma\})$ be an ADP.  When max-greedy policies exist, we let
$\tmax$ be the Bellman max-operator and $\Hmax$ be the Howard max-operator.
As above, we denote the greatest element of $V_\Sigma$ by $\vmax$ whenever it exists.

\begin{lemma}\label{l:hmu}
    If every $v \in V$ has at least one max-greedy policy, then the following
    statements are true:
    \begin{enumerate}
        \item $\Hmax$ obeys $v_\sigma \preceq \Hmax \, v_\sigma$ for all $\sigma \in \Sigma$.
        \item If $\sigma \in \Sigma$ and $T v_\sigma = v_\sigma$, then $\vmax$
            exists in $V$ and $v_\sigma = \vmax$.
        \item If $v \in V$ and $\Hmax \, v =v$, then $v = \vmax$ and $\tmax \, \vmax = \vmax$.
        \item If $v \in V$ and $\Sigma$ is finite, then $\vmax$ exists,
            $\Hmax \, \vmax = \vmax$ and 
            $(\Hmax^k v)_{k\geq 0}$ converges to $\vmax$ in finitely many steps.
    \end{enumerate}
\end{lemma}

\begin{proof}
   As for (i), fix $\sigma \in \Sigma$ and let $\tau$ be such that $\Hmax  v_\sigma =
    v_{\tau}$. Since $\tau$ is $v_\sigma$-greedy, we have $v_\sigma = T_\sigma
    \, v_\sigma \preceq \tmax \, v_\sigma = T_\tau \, v_\sigma $. Upward stability of
    $T_\tau$ gives $v_{\sigma} \preceq v_\tau = \Hmax  v_\sigma$.

    As for (ii),
    suppose $\sigma \in \Sigma$ and $\tmax \, v_\sigma = v_\sigma$.
    Fix $\tau \in \Sigma$ and note that $v_\sigma = \tmax \, v_\sigma \succeq T_\tau
    \, v_\sigma$. Downward stability of $T_\tau$ implies $v_\sigma
    \succeq v_\tau$.  Since $\tau \in \Sigma$ was arbitrary, $v_\sigma =
    \vmax$. 
    
    As for (iii), fix $v \in V$ with $\Hmax \, v = v$ and let $\sigma$ be such that
    $\Hmax \, v = v_\sigma$.  Then $v_\sigma = v$, and, 
    since $\sigma$ is $v$-max-greedy, $T_\sigma \, v = \tmax \, v$.  But then
    $T_\sigma \, v_\sigma = \tmax \, v_\sigma$, and,
    since $v_\sigma = T_\sigma \, v_\sigma$, we have $v_\sigma = \tmax \,
    v_\sigma$. Part (ii) now implies $v = v_\sigma =
    \vmax$.   This proves the first claim.  Regarding the second, substituting
    $v_\sigma = \vmax$ into $v_\sigma = \tmax \, v_\sigma$ yields $\vmax = \tmax
    \,\vmax$.

    For (iv), it suffices to show that $\Hmax \,  \vmax = \vmax$ and there
    exists a $K \in \NN$ such that $\Hmax^K v = \vmax$.  To this end,
    let $v_k = \Hmax^k v$ and note that $v_k \in V_\Sigma$ for all $k \geq 1$.
    Part (i) implies that $v_{k+1} \succeq v_k$ for all $k \in \NN$.
    Since the sequence $(v_k)$ is contained in the finite set
    $V_\Sigma$, it must be that $v_{K+1} = v_K$ for
    some $K \in \NN$ (since otherwise $V_\Sigma$ contains an infinite sequence
    of distinct points).
    But then $\Hmax \, v_K = v_{K+1} = v_K$, so $v_K$ is a fixed point of
    $\Hmax$. Part (iii) now implies that $v_K =\vmax$.
\end{proof}

\begin{proof}[Proof of Proposition~\ref{p:fposet}]\label{pp:fposet}
    If $\aA$ is an ADP such that max-greedy policies exist
    and $\Sigma$ is finite, then, by (iii)--(iv) of Lemma~\ref{l:hmu},
    the point $\vmax$ is a fixed point of $\tmax$.  This proves the max version of
    Proposition~\ref{p:fposet}.  The proof of the min version is analogous.
\end{proof}

\begin{lemma}\label{l:fpeibe}
    If $\aA$ is max-stable, then the following statements hold.
    \begin{enumerate}
        \item $V_\Sigma$ has a greatest element $\vmax$ and
        \item $\vmax$ is the unique fixed point of $\tmax$ in $V$.
        \item a policy is max-optimal if and only if it is $\vmax$-max-greedy.
        \item at least one optimal policy exists.
    \end{enumerate}
\end{lemma}

\begin{proof}
    As for parts (i)--(ii), we observe that,
    by max-stability, $\tmax$ has a fixed point $\bar v$ in $V$.
    By existence of max-greedy policies, we can find a $\sigma \in \Sigma$
    such that $\bar v = \tmax \, \bar v = T_\sigma \, \bar v$.  But $T_\sigma$
    has a unique fixed point in $V$, equal to $v_\sigma$, so $\bar v =
    v_\sigma$.  Moreover, if $\tau$ is any policy, then $T_\tau \,
    \bar v \preceq \tmax \, \bar v = \bar v$ and hence, by downward stability,
    $v_\tau \preceq \bar v$.  These facts imply that $\vmax := \bar v$ is the greatest
    element of $V_\Sigma$ and a fixed point of $\tmax$.
    Since greatest elements are unique, $\vmax$ is the only fixed point of
    $\tmax$ in $V$.

    For (iii), parts (i)--(ii) give 
     $\vmax \in V$ and $\tmax \, \vmax = \vmax$.  Now
     recall that $\sigma$ is optimal if and only if $v_\sigma = \vmax$.
    Since $v_\sigma$ is the unique fixed point of $T_\sigma$, this is
    equivalent to $T_\sigma \, \vmax = \vmax$.  Since $\tmax \, \vmax = \vmax$, the last
    statement is equivalent to $T_\sigma \, \vmax = \tmax \vmax$, which is, in turn
    equivalent to the statement that $\sigma$ is $\vmax$-greedy.

    Part (iv) follows from part (iii) and existence of a $\vmax$-greedy policy. 
\end{proof}

\begin{proof}[Proof of Theorem~\ref{t:fbk}]
    Parts (i)--(iv) of Theorem~\ref{t:fbk} follow from Lemma~\ref{l:fpeibe}.
    The last claim follows from Lemma~\ref{l:hmu}.
\end{proof}

%
% \begin{example}\label{eg:subcannot}
%     Let $V = \{a, b\}$ and $\hat V = \{\hat a, \hat b\}$ where $a \preceq b$ and
%     $\hat a \preceq \hat b$.  Let
%     %
%     \begin{equation*}
%         F \, a = F \, b = \hat a,
%         \quad
%         G_1 \, \hat a = G_1 \, \hat b = b
%         \quad \text{and} \quad
%         G_2 \, \hat a = G_2 \, \hat b = a.
%     \end{equation*}
%     %
%     Being constant functions, $F$, $G_1$ and $G_2$ are order-preserving.
%     Setting $T_i := G_i \circ F$ and $\hat T_i = F \circ G_i$ forms ADPs
%     $\aA = (V, \{T_i\})$ and $\hat \aA = (\hat V, \{\hat T_i\})$. 
%     By construction, $\hat \aA$ is subordinate to $\aA$.  For $\aA$ we
%     have $T_1 a = T_1 b = b$, so $T_1$ is order stable and $v_1 = b$. Similarly,
%     $T_2$ is order stable and $v_2 = a$.  Hence $\aA$ is order stable and $\vmax
%     = v_1 = b$.  Letting $\Sigma^*$ be the set of optimal policies, we have
%     $\Sigma^* = \{1\}$.  For $\hat \aA$ we have $\hat T_i \, \hat a = \hat T_i
%     \, \hat b = \hat a$ for $i=1,2$, so $\hat v_1 = \hat v_2 = \hat a$.  Hence
%     $\hvmax = \hat a$ and the set of optimal policies is $\hat \Sigma^* = \{1,
%     2\}$.  This shows that the reverse implication in the final part of
%     Theorem~\ref{t:rgsc} is invalid:  the fact that $\sigma$ is max-optimal for
%     $\hat \aA$ does not imply that if $\sigma$ is max-optimal for $\aA$.
% \end{example}

\subsection{Proofs of Isomorphism Results}

\begin{proof}[Proof of Theorem~\ref{t:iso}]\label{pf:iso}
    Claims (i)--(ii) follow directly from Lemma~\ref{l:ocos}.  Regarding (iii),
    suppose $\aA$ is max-stable.  We claim that, for $\hat
    \aA$, max-greedy policies always exist. To see this, fix $\hat v \in \hat
    V$. Since $\aA$ is max-stable, we can choose $\sigma$ to be $F^{-1} \hat
    v$-max-greedy, so that $T_\tau \, F^{-1} \hat v \preceq T_\sigma \, F^{-1} \hat v$
    for all $\tau \in \Sigma$.  Then $F^{-1} \hat T_\tau \, \hat v \preceq F^{-1}
    \hat T_\sigma \, \hat v$ and hence  $\hat T_\tau \, \hat v \preceq \hat T_\sigma
    \, \hat v$ for all $\tau \in \Sigma$. In particular, $\sigma$ is $\hat
    v$-max-greedy.

    Continuing to assume that $\aA$ is max-stable, we now prove
    \eqref{eq:isotmax}.  For given $v \in V$, applying
    the order conjugacy \eqref{eq:adpih} yields
    \begin{equation*}
        \tmax v
        = \bigvee_\sigma T_\sigma \, v
        = \bigvee_\sigma F^{-1} \hat T_\sigma \, F \, v
        = F^{-1} \bigvee_\sigma \hat T_\sigma \, F \, v
        = F^{-1} \htmax \, F \, v,
    \end{equation*}
    which is equivalent to $F \circ \tmax = \htmax \circ F$ from
    \eqref{eq:isotmax}, and implies that $(V, \tmax)$ and $(\hat V, \htmax)$ are
    order conjugate under $F$. By max-stability of $\aA$ and Theorem~\ref{t:fbk}, the operator $\tmax$
    has unique fixed point $\vmax$ in $V$. Lemma~\ref{l:ocos} then implies that
    $\htmax$ has unique fixed point $F \vmax$ in $\hat V$.  This completes the
    proof that $\hat \aA$ is max-stable.  By max-stability of $\hat \aA$, the
    unique fixed point of $\htmax$ in $\hat V$ is $\hvmax$, so $F \vmax =
    \hvmax$ and both claims in \eqref{eq:isotmax} are verified. Finally, note that $\sigma$ is max-optimal for $\aA$ if and only if
    $T_\sigma \, \vmax = \vmax$, which, by the bijection property of $F$, is
    also equivalent to $F \, T_\sigma \, \vmax = \hvmax$.  Using $F \circ T_\sigma
    = \hat T_\sigma \circ F$, we can write this as $\hat T_\sigma \, \hvmax = \hvmax$,
    which is equivalent to the statement that $\sigma$ is max-optimal for $\hat
    \aA$. We have now confirmed all the claims in (iii).

    The proof of (iv) is identical after replacing max with min and $\bigvee$
    with $\bigwedge$.  (Alternatively, the proof can be derived from (iii) and duality.)
\end{proof}

\begin{proof}[Proof of Theorem~\ref{t:antiiso}]\label{pf:antiiso}
    Let $\aA$ and $\hat \aA$ be anti-isomorphic, so that $\aA$ is
    isomorphic to $\hat \aA^\partial$.  If $\aA$ is well-posed,
    then, by Theorem~\ref{t:iso}, $\hat \aA^\partial$ is well-posed,
    so $\hat T_\sigma$ has a unique fixed point in $\hat V$ for all $\sigma \in
    \Sigma$.  This implies that $\hat \aA$ is likewise well-posed, completing
    the proof of (i).  Similarly, if $\aA$ is  order stable, then, by
    Theorem~\ref{t:iso}, $\hat \aA^\partial$ is order stable, in which case
    $\hat \aA$ is  order stable, by Lemma~\ref{l:odod}. This proves (ii).

    Now suppose $\aA$ is max-stable.  Then, by $\aA \sim \hat \aA^\partial$ and
    Theorem~\ref{t:iso}, $\hat \aA^\partial$ is max-stable with $F \circ \tmax =
    \htmax^\partial \circ F$ and $\hvmax^\partial = F \, \vmax$.  As with our
    discussion of duality in Appendix~\ref{app:minim}, this is equivalent to $F \circ
    \tmax = \htmin \circ F$ and $\hvmin = F \, \vmax$, which proves
    \eqref{eq:antiisotmax}.

    Finally, Theorem~\ref{t:iso} tells us that $\aA$ and $\hat \aA^\partial$
    have the same max-optimal policies.  Applying Lemma~\ref{l:mmp},
    we see that the max-optimal policies of $\aA$ are the same as the
    min-optimal policies of $\hat \aA$.
\end{proof}

\subsection{Min-Optimality}\label{app:minim}

In the abstract setting of Section~\ref{s:aop}, minimization results are readily recovered from
maximization results by order duality.

We call a policy $\sigma \in \Sigma$ \navy{min-optimal} for $\aA$ if
$v_\sigma$ is a least element of $V_\Sigma$. When $V_\Sigma$ has a least element
we denote it by $\vmin$ and call it the \navy{min-value function}.
We define $\Hmin$ from $V$ to $\{v_\sigma\}$ via
$\Hmin \, v = v_\sigma$  where $\sigma$ is $v$-min-greedy and call $\Hmin$ the
\navy{Howard min-operator} generated by $\aA$.

Below, if $\aA := (V, \{T_\sigma\})$ is an ADP then its \navy{dual}
$\aA^\partial$ is the ADP $(V^\partial, \{T_\sigma\})$ where the partial order
$\preceq$ on $V$ is replaced with its dual $\preceq^\partial$.
In this setting, we let $\tmax^\partial$ be the Bellman max-operator for
$\aA^\partial$, $\vmax^\partial$ be the max-value function for $\aA^\partial$,
and so on.

\begin{lemma}\label{l:mmp}
    $\aA$ is min-stable if and only if $\aA^\partial$ is max-stable, in which
    case $\tmin = \tmax^\partial$ and $\Hmin =
    \Hmax^\partial$. A policy $\sigma$ is max-optimal for $\aA$ if and only if
    $\sigma$ is min-optimal for $\aA^\partial$.
\end{lemma}

\begin{proof}
    Let $\aA$ be min-stable.  Then $\aA^\partial$ is order stable, by
    Lemma~\ref{l:odod}. Now fix $v \in V$ and suppose that $\sigma$ is
    min-greedy for $\aA$, so that $T_\sigma \, v \preceq T_\tau \, v$ for all
    $\tau \in \Sigma$.  Then $T_\sigma \, v \succeq^\partial T_\tau \, v$ for
    all $\tau \in \Sigma$, so $\sigma$ is $v$-max-greedy for $\aA^\partial$
    and $\tmax^\partial v = T_\sigma v = \tmin v$.
    We have proved that $\aA^\partial$ is max-stable and $\tmax^\partial =
    \tmin$. The remaining steps follow easily from the definitions.
\end{proof}

Results analogous to Theorem~\ref{t:fbk} hold for  minimization.

\begin{theorem}[Min-optimality]\label{t:fbk_min}
    If $\aA$ is min-stable, then
    \begin{enumerate}
        \item the min-value function $\vmin$ exists in $V$,
        \item $\vmin$ is the unique solution to the Bellman min-equation in $V$,
        \item a policy is min-optimal if and only if it is $\vmin$-min-greedy.
        \item at least one min-optimal policy exists.
    \end{enumerate}
    If, in addition, $\Sigma$ is finite, then Howard min-policy iteration
    converges to $\vmin$ in finitely many steps.
\end{theorem}

\begin{proof}[Proof of Theorem~\ref{t:fbk_min}]\label{pf:fbk_min}
    Let $\aA$ be min-stable.  By
    Lemma~\ref{l:mmp}, the dual $\aA^\partial$ is max-stable.  Hence, by
    Theorem~\ref{t:fbk}, $\vmax^\partial$ exists in $V$.
    But then $\vmin$ exists in $V$ and is equal to $\vmax^\partial$,
    since $\vmin = \bigwedge_\sigma v_\sigma = \bigvee_\sigma^\partial v_\sigma
    = \vmax^\partial$. Also, by Theorem~\ref{t:fbk}, $\vmax^\partial$ is the unique
    solution to $\tmax^\partial \vmax^\partial = \vmax^\partial$.
    Applying Lemma~\ref{l:mmp}, we see that $\tmin \, \vmin =
    \vmin$.  The remaining claims follow by analogous arguments.
\end{proof}

\subsection{Proofs of Epstein--Zin Optimality Results}\label{ss:pezr}

This section offers a proof of Theorem~\ref{t:eziff}.
To do so, we establish the following.
\begin{enumerate}
    \item[(C1)] If $\eE(\beta, Q, \theta)^{1/\theta} < 1$ and $\gamma < 0$, then $\hat
        \aA$ is min-stable.
    \item[(C2)] If $\eE(\beta, Q, \theta)^{1/\theta} < 1$ and $\gamma > 0$, then $\hat
        \aA$ is max-stable.
    \item[(C3)] If $\eE(\beta, Q, \theta)^{1/\theta} \geq 1$, then $\hat \aA$ is not well-posed.
\end{enumerate}
Together these facts establish Theorem~\ref{t:eziff}.  Indeed, if (C1)
holds, then, since $\aA$ and $\hat \aA$ are anti-isomorphic (see
Lemma~\ref{l:f}), it follows that $\aA$ is max-stable (Theorem~\ref{t:antiiso}).
If (C2) holds, then, since $\aA$ and $\hat \aA$ are isomorphic (see
Lemma~\ref{l:f}), it follows that $\aA$ is max-stable (Theorem~\ref{t:iso}).
Finally, if (C3) holds, then 
$\aA$ is also not well-posed (by Theorem~\ref{t:iso} or
Theorem~\ref{t:antiiso}, depending on whether $\gamma > 0$ or $\gamma < 0$.)
When $\aA$ is not well-posed, recursive utility does not exist, so the dynamic program is
undefined.

In what follows, given $\sigma \in \Sigma$ we set
\begin{equation*}
    A_\sigma(x, x') := \beta(x)^\theta P(x, \sigma(x), x')
    \qquad (x, x' \in \Xsf).
\end{equation*}
Also, for any linear operator $B$, the symbol $\rho(B)$ represents the spectral
radius.

\begin{lemma}\label{l:soz}
    For all $\sigma \in \Sigma$, we have $\rho(A_\sigma) = \eE(\beta, Q, \theta)$.
\end{lemma}

\begin{proof}
    Fix $z \in \Zsf$ and let $\1$ be a vector of ones.  An inductive argument
    shows that 
    \begin{equation}\label{eq:akz}
        (A_\sigma^k \1)(x) 
        = (A_\sigma^k \1)(z) 
        = \EE \prod_{t=0}^{k-1} \beta(Z_t(z))^\theta.
    \end{equation}
    Combining \eqref{eq:akz} with Theorem~9.1 of
    \cite{krasnoselskii1972approximate}, we have
    \begin{equation*}
        \rho(A_\sigma) 
        = \lim_{k \to \infty} \left\{ \sup_z ( A_\sigma^k \1)(z) \right\}^{1/k}
        = \lim_{k \to \infty}
            \left\{
                \sup_{z \in \Zsf}
                \EE \prod_{t=0}^{k-1} \beta(Z_t(z))^\theta
            \right\}^{1/k},
    \end{equation*}
    as was to be shown.
\end{proof}

\begin{lemma}\label{l:ezos}
    The ADP $\hat \aA$ is order stable if and only if $\eE(\beta, Q, \theta)^{1/\theta} < 1$.
    Moreover, if this condition fails, then $\hat \aA$ is not well
    posed.
\end{lemma}

\begin{proof}
    Fix $\sigma \in \Sigma$ and let $V$ and $\hat T_\sigma$ be 
    as defined in Section~\ref{ss:modep}. By Theorem~3.1 of \cite{stachurski2025unique}, 
    \begin{enumerate}
        \item $\rho(A_\sigma)^{1/\theta} < 1$ $\implies$ $(V, \hat T_\sigma)$ is
            globally stable on $V$, and 
        \item $\rho(A_\sigma)^{1/\theta} \geq 1$ 
            $\implies$ $\hat T_\sigma$ has no fixed point in $V$.
    \end{enumerate}
    We saw in Lemma~\ref{l:soz} that $\rho(A_\sigma) = \eE(\beta, Q, \theta)$, 
    so $\eE(\beta, Q, \theta)^{1/\theta} < 1$ if and only if (i) holds.
    In this case, $(V, \hat T_\sigma)$ is globally stable and hence order
    stable (by Example~\ref{eg:pspace}).  Therefore $\hat \aA$ is order stable.

    If, on the other hand, $\eE(\beta, Q, \theta)^{1/\theta} \geq 1$, then
    (ii) holds and $\hat \aA$ is not well-posed (and therefore not order stable).
\end{proof}

Now we return to (C1)--(C3) above.  Assume the conditions in (C1).
Then $\hat \aA$ is order stable by Lemma~\ref{l:ezos}.  Also, for
$v \in V$, we construct a $v$-min-greedy policy $\sigma$ by taking
\begin{equation*}
    \sigma(x)
    \in \argmin
    \left\{
        r(x, a)^\alpha + \beta (x)
            \left[
                \sum_{x' \in \Xsf} v(x') P(x,a, x')
            \right]^{1/\theta}
    \right\}^{\theta}
\end{equation*}
for all $x \in \Xsf$.  Since the policy set is finite,
Proposition~\ref{p:fposet} implies that $\hat \aA$ is min-stable. Hence (C1)
holds.  The proof of (C2) is analogous.  Finally, (C3) follows directly from
Lemma~\ref{l:ezos}.

%\begin{proof}
%    Suppose $\bar v \in \fix(T)$. 
%    By existence of greedy policies, there exists a $\sigma \in \Sigma$
%    satisfying $T \bar v = T_{\sigma} \bar v$.  Since
%    $ \bar v = T  \bar v = T_{\sigma}  \bar v$ and $v_{\sigma}$ is the only
%    fixed point of $T_{\sigma}$, we have $\bar v = v_{\sigma}$.  Hence
%    $\bar v \preceq \vmax$.
%    To check the reverse inequality, fix $\sigma \in \Sigma$,
%    and note that $ \bar v = T  \bar v \succeq T_{\sigma}  \bar v$.  Downward
%    stability now yields $ \bar v \succeq v_{\sigma}$. Since $\sigma$ is
%    arbitrary it follows that $\bar v \succeq \vmax$.  Therefore $\bar v =
% \vmax$.
%    Since every fixed point of $T$ equals $\vmax$, we see that $\vmax$ is the unique
%    solution to the Bellman equation in $V$.
%
%    To see that Bellman's principle of optimality
%    holds, recall that $\sigma$ is optimal if and only if $v_\sigma = \vmax$.
%    Since $v_\sigma$ is the unique fixed point of $T_\sigma$, this is
%    equivalent to $T_\sigma \, \vmax = \vmax$.  Since $T\vmax = \vmax$, the last
%    statement is equivalent to $T_\sigma \, \vmax = T \vmax$, which is, in turn
%    equivalent to the statement that $\sigma$ is $\vmax$-greedy.  This proves
%    claim (iii) in the theorem.
%
%    Regarding claim (iv), since Bellman's principle of optimality holds and
%    there exists at least one
%    $\vmax$-greedy policy, an optimal policy also exists.
%\end{proof}

\bibliographystyle{apalike}
\bibliography{qe_bib}

@inproceedings{Farahmand2010Error,
  author    = {Farahmand, Amir-Massoud and Munos, R{\'e}mi  and Szepesv{\'a}ri, Csaba},
  title     = {Error Propagation for Approximate Policy and Value Iteration},
  booktitle = {Advances in Neural Information Processing Systems 23},
  editor    = {John D. Lafferty and Christopher K. I. Williams and John Shawe-Taylor and Richard S. Zemel and Aron Culotta},
  pages     = {568--576},
  year      = {2010},
  url       = {http://papers.neurips.cc/paper/4181-error-propagation-for-approximate-policy-and-value-iteration.pdf}
}

@article{sargent2025partially,
      title={Dynamic Programming on Partially Ordered Sets}, 
      journal={SIAM Journal on Control and Optimization},
      author={Thomas J. Sargent and John Stachurski},
      year={2025},
      volume={in press}
}

@book{arnold2012geometrical,
  title={Geometrical methods in the theory of ordinary differential equations},
  author={Arnold, Vladimir Igorevich},
  volume={250},
  year={2012},
  publisher={Springer Science \& Business Media}
}

@article{smale1967differentiable,
  title={Differentiable dynamical systems},
  author={Smale, Stephen},
  journal={Bulletin of the American Mathematical Society},
  volume={73},
  number={6},
  pages={747--817},
  year={1967}
}

@article{bauerle2018stochastic,
  title={Stochastic optimal growth model with risk sensitive preferences},
  author={B{\"a}uerle, Nicole and Ja{\'s}kiewicz, Anna},
  journal={Journal of Economic Theory},
  volume={173},
  pages={181--200},
  year={2018},
  publisher={Elsevier}
}

@misc{fedus2019hyperbolic,
      title={Hyperbolic Discounting and Learning over Multiple Horizons}, 
      author={William Fedus and Carles Gelada and Yoshua Bengio and Marc G. Bellemare and Hugo Larochelle},
      year={2019},
      eprint={1902.06865},
      archivePrefix={arXiv},
      primaryClass={stat.ML}
}

@book{krasnoselskii1972approximate,
  author	= {Krasnosel’skii, M. A. and Vainikko, G. M. and Zabreiko,
		  P. P. and Rutitskii, Ya. B. and Stetsenko, V. Ya.},
  title		= {Approximate Solution of Operator Equations},
  year		= 1972,
  doi		= {10.1007/978-94-010-2715-1},
  url		= {http://dx.doi.org/10.1007/978-94-010-2715-1},
  isbn		= 9789401027151,
  publisher	= {Springer Netherlands}
}

@article{howard1972risk,
  title={Risk-sensitive {M}arkov decision processes},
  author={Howard, Ronald A and Matheson, James E},
  journal={Management Science},
  volume={18},
  number={7},
  pages={356--369},
  year={1972},
  publisher={INFORMS}
}

@book{layek2015introduction,
  title={An Introduction to Dynamical Systems and Chaos},
  author={Layek, G.C.},
  isbn={9788132225560},
  year={2015},
  publisher={Springer}
}

@book{sternberg2014dynamical,
  title={Dynamical Systems},
  author={Sternberg, S.},
  isbn={9780486135144},
  year={2014},
  publisher={Dover Publications}
}

@article{fei2021exponential,
  title={Exponential {B}ellman equation and improved regret bounds for
         risk-sensitive reinforcement learning},
  author={Fei, Yingjie and Yang, Zhuoran and Chen, Yudong and Wang, Zhaoran},
  journal={Advances in Neural Information Processing Systems},
  volume={34},
  pages={20436--20446},
  year={2021}
}

@book{davey2002introduction,
  title={Introduction to lattices and order},
  author={Davey, Brian A and Priestley, Hilary A},
  year={2002},
  publisher={Cambridge University Press}
}

@article{maliar2021deep,
  title={Deep learning for solving dynamic economic models.},
  author={Maliar, Lilia and Maliar, Serguei and Winant, Pablo},
  journal={Journal of Monetary Economics},
  volume={122},
  pages={76--101},
  year={2021},
  publisher={Elsevier}
}

@article{azinovic2022deep,
  title={Deep equilibrium nets},
  author={Azinovic, Marlon and Gaegauf, Luca and Scheidegger, Simon},
  journal={International Economic Review},
  volume={63},
  number={4},
  pages={1471--1525},
  year={2022},
  publisher={Wiley Online Library}
}

@article{gomez2020important,
	author = {Gomez-Cram, Roberto and Yaron, Amir},
	doi = {10.1093/rfs/hhaa039},
	issn = {0893-9454},
	journal = {The Review of Financial Studies},
	month = {05},
	number = {2},
	pages = {985-1045},
	title = {{How Important Are Inflation Expectations for the Nominal Yield Curve?}},
	volume = {34},
	year = {2020},
}

@article{weil1990nonexpected,
  title={Nonexpected utility in macroeconomics},
  author={Weil, Philippe},
  journal={The Quarterly Journal of Economics},
  volume={105},
  number={1},
  pages={29--42},
  year={1990},
  publisher={MIT Press}
}

@article{de2022dynamic,
  title={Dynamic Economics with Quantile Preferences},
  author={de Castro, Luciano and Galvao, Antonio F and Nunes, Daniel},
  journal={Available at SSRN 4108230},
  year={2022}
}

@article{kristensen2021solving,
  title={Solving dynamic discrete choice models using smoothing and sieve methods},
  author={Kristensen, Dennis and Mogensen, Patrick K and Moon, Jong Myun and Schjerning, Bertel},
  journal={Journal of Econometrics},
  volume={223},
  number={2},
  pages={328--360},
  year={2021},
  publisher={Elsevier}
}

@article{stachurski2025unique,
  title={Unique solutions to power-transformed affine systems},
  author={Stachurski, John and Wilms, Ole and Zhang, Junnan},
  journal={Journal of Mathematical Analysis and Applications},
  volume={550},
  number={1},
  pages={129515},
  year={2025},
  publisher={Elsevier}
}

@article{bloise2024not,
  title={Do not Blame {B}ellman: It Is {K}oopmans' Fault},
  author={Bloise, Gaetano and Le Van, Cuong and Vailakis, Yiannis},
  journal={Econometrica},
  volume={92},
  number={1},
  pages={111--140},
  year={2024},
  publisher={Wiley Online Library}
}

@book{hernandez2012discrete,
  title     = {Discrete-time {M}arkov control processes: basic optimality criteria},
  author    = {Hern{\'a}ndez-Lerma, On{\'e}simo and Lasserre, Jean B},
  volume    = {30},
  year      = {2012},
  publisher = {Springer Science \& Business Media}
}

@article{bloise2018convex,
  title     = {Convex dynamic programming with (bounded) recursive utility},
  author    = {Bloise, Gaetano and Vailakis, Yiannis},
  journal   = {Journal of Economic Theory},
  volume    = {173},
  pages     = {118--141},
  year      = {2018},
  publisher = {Elsevier}
}

@article{degroot2018,
  author  = {de Groot, Oliver and Richter, 
             Alexander W. and Throckmorton, Nathaniel A.},
  journal = {Econometrica},
  number  = {4},
  pages   = {1513-1526},
  title   = {Uncertainty Shocks in a Model of Effective Demand: Comment},
  volume  = {86},
  year    = {2018}
}

@article{albuquerque2016valuation,
  title     = {Valuation risk and asset pricing},
  author    = {Albuquerque, Rui and Eichenbaum, Martin and Luo, Victor Xi and Rebelo, Sergio},
  journal   = {The Journal of Finance},
  volume    = 71,
  number    = 6,
  pages     = {2861--2904},
  year      = 2016,
  publisher = {Wiley Online Library}
}

@article{stachurski2021dynamic,
  title     = {Dynamic programming with state-dependent discounting},
  author    = {Stachurski, John and Zhang, Junnan},
  journal   = {Journal of Economic Theory},
  volume    = {192},
  pages     = {105190},
  year      = {2021},
  publisher = {Elsevier}
}

@article{bauerle2021stochastic,
  title={Stochastic dynamic programming with non-linear discounting},
  author={B{\"a}uerle, Nicole and Ja{\'s}kiewicz, Anna and Nowak, Andrzej S},
  journal={Applied Mathematics \& Optimization},
  volume={84},
  number={3},
  pages={2819--2848},
  year={2021},
  publisher={Springer}
}

@book{bauerle2011markov,
  title     = {Markov decision processes with applications to finance},
  author    = {B{\"a}uerle, Nicole and Rieder, Ulrich},
  year      = {2011},
  publisher = {Springer Science \& Business Media}
}

@article{rust1994structural,
  title     = {Structural estimation of {M}arkov decision processes},
  author    = {Rust, John},
  journal   = {Handbook of Econometrics},
  volume    = {4},
  pages     = {3081--3143},
  year      = {1994},
  publisher = {Elsevier}
}

@article{denardo1967contraction,
  title     = {Contraction mappings in the theory underlying dynamic programming},
  author    = {Denardo, Eric V},
  journal   = {Siam Review},
  volume    = {9},
  number    = {2},
  pages     = {165--177},
  year      = {1967},
  publisher = {SIAM}
}

@article{balbus2022time,
  title     = {Time-consistent equilibria in 
               dynamic models with recursive payoffs and behavioral discounting},
  author    = {Balbus, {\L}ukasz and Reffett, Kevin and Wo{\'z}ny, {\L}ukasz},
  journal   = {Journal of Economic Theory},
  pages     = {105493},
  year      = {2022},
  publisher = {Elsevier}
}

@article{schorfheide2018identifying,
  title     = {Identifying long-run risks: A {B}ayesian mixed-frequency approach},
  author    = {Schorfheide, Frank and Song, Dongho and Yaron, Amir},
  journal   = {Econometrica},
  volume    = {86},
  number    = {2},
  pages     = {617--654},
  year      = {2018},
  publisher = {Wiley Online Library}
}

@techreport{gao2021robust,
  title       = {Robust risk-sensitive reinforcement learning agents for trading markets},
  author      = {Gao, Yue and Lui, Kry Yik Chau and Hernandez-Leal, Pablo},
  institution = {arXiv preprint arXiv:2107.08083},
  year        = {2021}
}

@article{cagetti2002robustness,
  title     = {Robustness and pricing with uncertain growth},
  author    = {Cagetti, Marco and Hansen, Lars Peter and Sargent, Thomas J and Williams, Noah},
  journal   = {The Review of Financial Studies},
  volume    = {15},
  number    = {2},
  pages     = {363--404},
  year      = {2002},
  publisher = {Oxford University Press}
}

@book{hansen2011robustness,
  title     = {Robustness},
  author    = {Hansen, Lars Peter and Sargent, Thomas J},
  year      = {2011},
  publisher = {Princeton university press}
}

@article{epstein1989risk,
  author    = {Epstein, Larry G. and Zin, Stanley E},
  journal   = {Econometrica},
  number    = {4},
  pages     = {937-969},
  publisher = {Econometric Society},
  title     = {Risk Aversion 
               and the Temporal Behavior of Consumption and Asset Returns: A
               Theoretical Framework},
  volume    = {57},
  year      = {1989}
}

@article{marinacci2019unique,
  title     = {Unique tarski fixed points},
  author    = {Marinacci, Massimo and Montrucchio, Luigi},
  journal   = {Mathematics of Operations Research},
  volume    = {44},
  number    = {4},
  pages     = {1174--1191},
  year      = {2019},
  publisher = {INFORMS}
}

@article{marinacci2010unique,
  author    = {Marinacci, Massimo and Montrucchio, Luigi},
  journal   = {Journal of Economic Theory},
  number    = {5},
  pages     = {1776-1804},
  publisher = {Science Direct},
  title     = {Unique solutions for stochastic recursive utilities},
  volume    = {145},
  year      = {2010}
}

@book{bellman1957dynamic,
  title     = {Dynamic programming},
  author    = {Bellman, Richard},
  journal   = {Science},
  year      = {1957},
  publisher = {American Association for the Advancement of Science}
}

@book{kochenderfer2022algorithms,
  title     = {Algorithms for decision making},
  author    = {Kochenderfer, Mykel J and Wheeler, Tim A and Wray, Kyle H},
  year      = {2022},
  publisher = {The MIT Press}
}

@article{long1983real,
  title     = {Real business cycles},
  author    = {Long, John B and Plosser, Charles I},
  journal   = {Journal of Political Economy},
  volume    = {91},
  number    = {1},
  pages     = {39--69},
  year      = {1983},
  publisher = {The University of Chicago Press}
}

@book{bertsekas2022abstract,
  title     = {Abstract dynamic programming},
  author    = {Bertsekas, Dimitri P},
  year      = {2022},
  edition   = {3},
  publisher = {Athena Scientific}
}

@book{bertsekas2012dynamic,
  title     = {Dynamic programming and optimal control},
  author    = {Bertsekas, Dimitri},
  volume    = {1},
  year      = {2012},
  publisher = {Athena Scientific}
}

@book{stokey1989recursive,
  title     = {Recursive methods in dynamic economics},
  author    = {Stokey, Nancy L and Lucas,  Robert E},
  year      = {1989},
  publisher = {Harvard University Press}
}

@book{puterman2005markov,
  title     = {Markov decision processes: discrete stochastic dynamic programming},
  author    = {Puterman, Martin L},
  year      = {2005},
  publisher = {Wiley Interscience}
}

@article{kingma2014adam,
  title={Adam: A method for stochastic optimization},
  author={Kingma, Diederik P and Ba, Jimmy},
  journal={arXiv preprint arXiv:1412.6980},
  year={2014}
}

@article{kennedy2008chaotic,
  title={Chaotic equilibria in models with backward dynamics},
  author={Kennedy, Judy A and Stockman, David R},
  journal={Journal of Economic Dynamics and Control},
  volume={32},
  number={3},
  pages={939--955},
  year={2008},
  publisher={Elsevier}
}

@article{gardini2009forward,
  title={Forward and backward dynamics in implicitly defined overlapping generations models},
  author={Gardini, Laura and Hommes, Cars and Tramontana, Fabio and De Vilder, Robin},
  journal={Journal of Economic Behavior \& Organization},
  volume={71},
  number={2},
  pages={110--129},
  year={2009},
  publisher={Elsevier}
}

@article{raines2012fixed,
  title={Fixed points imply chaos for a class of differential inclusions that arise in economic models},
  author={Raines, Brian and Stockman, David},
  journal={Transactions of the American Mathematical Society},
  volume={364},
  number={5},
  pages={2479--2492},
  year={2012}
}

@book{flynn2022macroeconomics,
  title={The macroeconomics of narratives},
  author={Flynn, Joel P and Sastry, Karthik},
  year={2022},
  publisher={SSRN}
}

@techreport{battaglini2021chaos,
  title={Chaos and unpredictability in dynamic social problems},
  author={Battaglini, Marco},
  year={2021},
  institution={National Bureau of Economic Research}
}

@article{deng2022continuous,
  title={Continuous unimodal maps in economic dynamics: On easily verifiable conditions for topological chaos},
  author={Deng, Liuchun and Khan, M Ali and Mitra, Tapan},
  journal={Journal of Economic Theory},
  volume={201},
  pages={105446},
  year={2022},
  publisher={Elsevier}
}

@article{rincon2024existence,
  title={Existence and uniqueness of solutions to the Bellman equation in stochastic dynamic programming},
  author={Rinc{\'o}n-Zapatero, Juan Pablo},
  journal={Theoretical Economics},
  volume={19},
  number={3},
  pages={1223--1260},
  year={2024},
  publisher={Wiley Online Library}
}

@article{ma2022unbounded,
  title={Unbounded dynamic programming via the Q-transform},
  author={Ma, Qingyin and Stachurski, John and Toda, Alexis Akira},
  journal={Journal of Mathematical Economics},
  volume={100},
  pages={102652},
  year={2022},
  publisher={Elsevier}
}

@article{foerster2016perturbation,
  title={Perturbation methods for Markov-switching dynamic stochastic general equilibrium models},
  author={Foerster, Andrew and Rubio-Ram{\'\i}rez, Juan F and Waggoner, Daniel F and Zha, Tao},
  journal={Quantitative economics},
  volume={7},
  number={2},
  pages={637--669},
  year={2016},
  publisher={Wiley Online Library}
}

@article{bayer2020solving,
  title={Solving discrete time heterogeneous agent models with aggregate risk and many idiosyncratic states by perturbation},
  author={Bayer, Christian and Luetticke, Ralph},
  journal={Quantitative Economics},
  volume={11},
  number={4},
  pages={1253--1288},
  year={2020},
  publisher={Wiley Online Library}
}

@article{brumm2017using,
  title={Using adaptive sparse grids to solve high-dimensional dynamic models},
  author={Brumm, Johannes and Scheidegger, Simon},
  journal={Econometrica},
  volume={85},
  number={5},
  pages={1575--1612},
  year={2017},
  publisher={Wiley Online Library}
}

@article{carroll2006method,
  title={The method of endogenous gridpoints for solving dynamic stochastic optimization problems},
  author={Carroll, Christopher D},
  journal={Economics letters},
  volume={91},
  number={3},
  pages={312--320},
  year={2006},
  publisher={Elsevier}
}

@article{barillas2007generalization,
  title={A generalization of the endogenous grid method},
  author={Barillas, Francisco and Fern{\'a}ndez-Villaverde, Jes{\'u}s},
  journal={Journal of Economic Dynamics and Control},
  volume={31},
  number={8},
  pages={2698--2712},
  year={2007},
  publisher={Elsevier}
}

\end{document}